\documentclass[conference, onecolumn]{IEEEtran}
\ifCLASSINFOpdf
\else
\fi
\usepackage[cmex10]{amsmath}
\usepackage{fixltx2e}

\usepackage{stfloats}
\fnbelowfloat
\usepackage{amssymb}
\usepackage{amsthm}
\usepackage{algpseudocode}
\usepackage{algorithm}
\usepackage{xspace}
\usepackage{multirow}
\usepackage{rotating}
\usepackage[table]{xcolor}

\input{Qcircuit}
\renewcommand{\Qcircuit}[1][0em]{\xymatrix @*=<#1>}

\newcommand{\braket}[2]{\langle #1 | #2 \rangle}

\newcommand{\cC}{\mathcal{C}}
\newcommand{\cM}{\mathcal{M}}
\newcommand{\cA}{\mathcal{A}}
\newcommand{\cF}{\mathcal{F}}

\newcommand{\normform}{basis form}

\makeatletter
\renewcommand*\env@matrix[1][*\c@MaxMatrixCols c]{%
  \hskip -\arraycolsep
  \let\@ifnextchar\new@ifnextchar
  \array{#1}}
\makeatother

\newtheorem{theorem}{Theorem}
\newtheorem{lemma}[theorem]{Lemma}
\newtheorem{proposition}[theorem]{Proposition}
\newtheorem{corollary}[theorem]{Corollary}
\newtheorem{observation}[theorem]{Observation}

\theoremstyle{definition}
\newtheorem{definition}[theorem]{Definition}

\begin{document}
%
\title{Efficient Inner-product Algorithm \\ for Stabilizer States}



%
%
\author{\IEEEauthorblockN{H\'{e}ctor J. Garc\'{i}a\IEEEauthorrefmark{1} \hspace{25pt} 
Igor L. Markov\IEEEauthorrefmark{1} \hspace{25pt} Andrew W. Cross\IEEEauthorrefmark{2}}
\IEEEauthorblockA{\small \hspace{-15pt} hjgarcia@eecs.umich.edu \hspace{10pt} imarkov@eecs.umich.edu  
\hspace{10pt} andrew.w.cross@ibm.com}
\IEEEauthorblockA{\ }
\IEEEauthorblockA{\IEEEauthorrefmark{1} \small University of Michigan -- EECS Department \\ 
2260 Hayward Street, Ann Arbor, MI, 48109-2121} 
\IEEEauthorblockA{\IEEEauthorrefmark{2} \small IBM T. J. Watson Research Center  \\ 
Yorktown Heights, NY 10598}
}


\maketitle

\thispagestyle{plain}
\pagestyle{plain}

\begin{abstract}

Large-scale quantum computation is likely to require massive quantum error correction (QEC). 
QEC codes and circuits are described via the stabilizer formalism, which represents
{\em stabilizer states} by keeping track of the operators that preserve them. Such states 
are obtained by {\em stabilizer circuits} (consisting of CNOT, Hadamard and Phase only)
and can be represented compactly on conventional computers 
using $\Omega (n^2)$ bits, where $n$ is the number of qubits \cite{Gottes98}. 
Although techniques for the efficient simulation of stabilizer circuits have been studied
extensively \cite{AaronGottes, Gottes, Gottes98}, techniques for efficient manipulation 
of stabilizer states are not currently available. To this end, we leverage
the theoretical insights from \cite{AaronGottes} and \cite{Vanden} to design new algorithms for: 
({\em i}) obtaining {\em canonical generators} for stabilizer states, ({\em ii}) obtaining  
{\em canonical stabilizer circuits}, and ({\em iii}) computing the inner product between stabilizer states.
Our inner-product algorithm takes $O(n^3)$ time in general, but observes
quadratic behavior for many practical instances relevant to QECC (e.g., GHZ states). 
We prove that each $n$-qubit stabilizer state has exactly $4(2^n - 1)$ {\em nearest-neighbor 
stabilizer states}, and verify this claim experimentally using our algorithms.
We design techniques for representing arbitrary quantum states
using {\em stabilizer frames} and generalize our algorithms to compute 
the inner product between two such frames.
\end{abstract}


%
\IEEEpeerreviewmaketitle

\section{Introduction} \label{sec:intro}

\noindent
Gottesman \cite{Gottes} and Knill showed that for certain types of non-trivial 
quantum circuits known as {\em stabilizer circuits}, efficient simulation on 
classical computers is possible. Stabilizer circuits are exclusively composed 
of {\em stabilizer gates}~--~Controlled-NOT, Hadamard and Phase gates 
(Figure~\ref{fig:chp_pauli}a)~--~followed by measurements in the
computational basis. Such circuits are applied to a computational basis state
(usually $\ket{00...0}$) and produce output states called {\em stabilizer states}.
The case of purely unitary stabilizer circuits
(without measurement gates) is considered often, e.g., by consolidating
measurements at the end. Stabilizer circuits can be simulated in
poly-time by keeping track of a set Pauli operators that stabilize\footnote{An 
operator $U$ is said to stabilize a state iff $U\ket{\psi}=\ket{\psi}$.} 
the quantum state. Such {\em stabilizer operators} 
uniquely represent a stabilizer state up to an unobservable global phase. 
Equation~\ref{eq:stab_count} shows that the number of $n$-qubit stabilizer states
grows as $2^{n^2/2}$, therefore, describing a generic stabilizer state
requires at least $n^2/2$ bits. Despite their compact representation, 
stabilizer states can exhibit multi-qubit entanglement and are
often encountered in many quantum information applications
such as Bell states, GHZ states, error-correcting codes and 
one-way quantum computation. To better understand the role stabilizer states 
play in such applications, researchers have designed techniques to quantify 
the amount of entanglement \cite{Fattal, Wunder, Hein}
in such states and studied relevant properties such as
purification schemes \cite{Dur}, Bell inequalities \cite{Guehne} and equivalence
classes \cite{Vanden04}. Efficient algorithms for the manipulation of stabilizer 
states (e.g., computing the angle between them), can help lead to additional 
insights related to linear-algebraic and geometric properties of stabilizer states.
	
In this work, we describe in detail algorithms for the efficient computation
of the inner product between stabilizer states. 
We adopt the approach outlined in \cite{AaronGottes}, which 
requires the synthesis of a unitary stabilizer circuit
that maps a stabilizer state to a computational basis state. 
The work in \cite{AaronGottes} shows that, for any unitary
stabilizer circuit, there exists an equivalent block-structured 
{\em canonical circuit} that applies a block of Hadamard ($H$) gates 
only, followed by a block of CNOT ($C$) only, then a block of 
Phase ($P$) gates only, and so on in the $7$-block sequence 
$H$-$C$-$P$-$C$-$P$-$C$-$H$. Using an alternate representation
for stabilizer states, the work in \cite{Vanden} proves the existence of a 
($H$-$C$-$P$-$CZ$)-canonical circuit, 
where the $CZ$ block consists of Controlled-$Z$ (CPHASE) gates.  
However, no algorithms are known to synthesize such smaller $4$-block
circuits given an arbitrary stabilizer state. In contrast,
we describe an algorithm for synthesizing ($H$-$C$-$CZ$-$P$-$H$)-canonical 
circuits given any input stabilizer state. We prove that any $n$-qubit 
stabilizer state $\ket{\psi}$ has exactly $4(2^n - 1)$ 
{\em nearest-neighbors} -- stabilizer states $\ket{\varphi}$ such
that $|\braket{\psi}{\varphi}|$ attains the largest possible value $\neq 1$.
Furthermore, we design techniques for representing arbitrary quantum states
using {\em stabilizer frames} and generalize our algorithms to compute 
the inner product between two such frames.

This paper is structured as follows. Section \ref{sec:background} reviews
the stabilizer formalism and relevant algorithms for 
manipulating stabilizer-based representations of quantum states.
Section \ref{sec:inprod_stab} describes our circuit-synthesis and inner-product 
algorithms. In Section \ref{sec:validate},
we evaluate the performance of our algorithms. Our 
findings related to geometric properties of stabilizer states
are described in Section \ref{sec:stabneighbors}.
In Section \ref{sec:stabframes}, we discuss stabilizer frames 
and how they can be used to represent arbitrary states and extend
our algorithms to compute the inner product between frames.
Section \ref{sec:conclude} closes with concluding remarks.

	\begin{figure*}[!t]\footnotesize\centering
	\vspace{-10pt}
	\begin{tabular}{cc}
    $
        H = \frac{1}{\sqrt{2}}\begin{pmatrix}
            1 & 1 \\
            1 & -1 \end{pmatrix} \quad
        P = \begin{pmatrix}
            1 & 0 \\
            0 & i \end{pmatrix} \quad
        CNOT = \begin{pmatrix}
            1 & 0 & 0 & 0 \\
            0 & 1 & 0 & 0 \\
            0 & 0 & 0 & 1 \\
            0 & 0 & 1 & 0 \end{pmatrix}
   $ & \hspace{40pt}
   $ X = \begin{pmatrix}
                0 & 1 \\
                1 & 0 \end{pmatrix} \quad
            Y = \begin{pmatrix}
                0 & -i \\
                i &  0 \end{pmatrix} \quad
            Z = \begin{pmatrix}
                1 & 0 \\
                0 & -1 \end{pmatrix}
    $ \\ 
    & \\
    \parbox{.40\linewidth}{\caption{\label{fig:chp_pauli} 
    (a) Unitary stabilizer gates Hadamard (H), Phase (P)}}
    & \hspace{20pt} Fig. 1.\hspace{6pt}(b) The Pauli matrices. \\
    and Controlled-NOT (CNOT). &
    \end{tabular}
 	\vspace{-10pt}
	\end{figure*}
	
\section{Background and Previous Work}  \label{sec:background}

Gottesman~\cite{Gottes98} developed a description for 
quantum states involving the {\em Heisenberg representation} often 
used by physicists to describe atomic phenomena. In this model,
one describes quantum states by keeping track of their symmetries 
rather than explicitly maintaining complex vectors. 
The symmetries are operators for which these states are $1$-eigenvectors.
Algebraically, symmetries form {\em group} structures,
which can be specified compactly by group generators.
It turns out that this approach, also known 
as the {\em stabilizer formalism}, can be used to represent an
important class of quantum states. 

\subsection{The stabilizer formalism} \label{sec:stab}

A unitary operator $U$ {\em stabilizes} a state $\ket{\psi}$ if
$\ket{\psi}$ is a $1$--eigenvector of $U$, i.e., $U\ket{\psi} 
= \ket{\psi}$ \cite{Gottes, NielChu}. We are interested in operators $U$
derived from the Pauli matrices shown in Figure~\ref{fig:chp_pauli}b 
The following table lists the one-qubit states stabilized
by the Pauli matrices.

\begin{center}
	\begin{tabular}{lccclc}
		$X$ : & $(\ket{0}+\ \ket{1})/\sqrt{2}$ &&& $-X$ : & $(\ket{0}-\ \ket{1})/\sqrt{2}$ \\
		$Y$ : & $(\ket{0}+i\ket{1})/\sqrt{2}$ &&&$-Y$ : & $(\ket{0}-i\ket{1})/\sqrt{2}$ \\
		$Z$ : & $\ket{0}$ &&& $-Z$ : & $\ket{1}$ \\
	\end{tabular}
\end{center}

Observe that $I$ stabilizes all states and $-I$ does not stabilize any state. 
As an example, the entangled state $(\ket{00} + \ket{11})/\sqrt{2}$ is stabilized by 
the Pauli operators $X\otimes X$, $-Y\otimes Y$, $Z\otimes Z$ and $I\otimes I$.  
As shown in Table~\ref{tab:pauli_mult}, it turns out that the Pauli matrices 
along with $I$ and the multiplicative factors $\pm1$, $\pm i$, form a 
{\em closed group} under matrix multiplication \cite{NielChu}. Formally, the {\em Pauli group} 
$\mathcal{G}_n$ on $n$ qubits consists of the $n$-fold tensor product 
of Pauli matrices, $P = i^kP_1\otimes\cdot\cdot\cdot\otimes P_n$ 
such that $P_j\in\{I, X, Y, Z\}$ and $k\in\{0,1,2,3\}$. For brevity, the tensor-product symbol 
is often omitted so that $P$ is denoted by a string of $I$, $X$, $Y$ 
and $Z$ characters or {\em Pauli literals} and a separate integer value
$k$ for the phase $i^k$. This string-integer pair representation allows us to compute 
the product of Pauli operators without explicitly 
computing the tensor products,\footnote{\scriptsize This holds true due to the 
identity: $(A\otimes B)(C \otimes D)=(AC\otimes BD)$.} ~e.g., 
$(-IIXI)(iIYII) = -iIYXI$. Since $\mid \mathcal{G}_n\mid= 4^{n+1}$, 
$\mathcal{G}_n$ can have at most $\log_2 \mid \mathcal{G}_n \mid = 
\log_2 4^{n+1} = 2(n + 1)$ irredundant generators \cite{NielChu}.
The key idea behind the stabilizer formalism is to represent an 
$n$-qubit quantum state $\ket{\psi}$ by its {\em stabilizer group} 
$S(\ket{\psi})$ -- the subgroup of $\mathcal{G}_n$ that stabilizes $\ket{\psi}$.
As the following theorem shows, if $|S(\ket{\psi})|=2^n$, the group uniquely specifies 
$\ket{\psi}$. 

	\begin{table}[!b]
        \centering
        \parbox{.45\linewidth}{ 
    		\caption{\label{tab:pauli_mult} Multiplication table for Pauli matrices.
    		Shaded cells indicate anticommuting products.}}
    		\vspace{-5pt}
    		\begin{center}
        \begin{tabular}{|c||c|c|c|c|}
            \hline
                &   $I$ & $X$                         & $Y$   & $Z$ \\ \hline\hline
            $I$ &   $I$ & $X$                         & $Y$   &  $Z$ \\ \hline
            $X$ &   $X$ & $I$   & \cellcolor[gray]{0.85} $iZ$  & \cellcolor[gray]{0.85} $-iY$ \\ \hline
            $Y$ &   $Y$ & \cellcolor[gray]{0.85} $-iZ$ & $I$   & \cellcolor[gray]{0.85} $iX$ \\ \hline
            $Z$ &   $Z$ & \cellcolor[gray]{0.85} $iY$  & \cellcolor[gray]{0.85} $-iX$ & $I$ \\
            \hline
        \end{tabular}
        \end{center}
    \end{table}

    \begin{theorem} \label{th:gen_commute}
        For an $n$-qubit pure state $\ket{\psi}$ and $k\leq n$, $S(\ket{\psi}) \cong {\mathbb Z}_2^k$. 
        If $k=n$, $\ket{\psi}$ is specified uniquely by $S(\ket{\psi})$ and is called a stabilizer state.
	\end{theorem}
	\begin{proof}
            (\emph{i}) To prove that $S(\ket{\psi})$ is commutative,
            let $P, Q \in S(\ket{\psi})$ such that $PQ\ket{\psi} = \ket{\psi}$. If $P$ and $Q$ anticommute,
            $-QP\ket{\psi} = -Q(P\ket{\psi}) = -Q\ket{\psi} = -\ket{\psi} \neq \ket{\psi}$.
            Thus, $P$ and $Q$ cannot both be elements of $S(\ket{\psi})$.
			
			\noindent
            ({\em ii}) To prove that every element of $S(\ket{\psi})$ is of
            degree $2$, let $P \in S(\ket{\psi})$ such that $P\ket{\psi} = \ket{\psi}$.
            Observe that $P^2 = i^lI$ for $l\in\{0,1,2,3\}$.
            Since $P^2\ket{\psi} = P(P\ket{\psi}) = P\ket{\psi} = \ket{\psi}$, we obtain
            $i^l = 1$ and $P^2 = I$.
			
			\noindent
            ({\em iii}) From group theory, a finite Abelian group with
            $a^2 = a$ for every element must be $\cong {\mathbb Z}_2^k$.
			
			\noindent
            ({\em iv}) We now prove that $k \leq n$.
            First note that each independent generator $P \in S(\ket{\psi})$
            imposes the linear constraint $P\ket{\psi}=\ket{\psi}$
            on the $2^n$-dimensional vector space.
            The subspace of vectors that satisfy such a constraint
            has dimension $2^{n-1}$, or half the space. Let $gen(\ket{\psi})$
            be the set of generators for $S(\ket{\psi})$.
            We add independent generators to $gen(\ket{\psi})$ one by one and impose
  			their linear constraints, to limit $\ket{\psi}$ to the shared
			$1$-eigenvector. Thus the size of $gen(\ket{\psi})$ is at most $n$.
			In the case $|gen(\ket{\psi})| = n$, the $n$ independent 
			generators reduce the subspace of possible states to dimension
        		one. Thus, $\ket{\psi}$ is uniquely specified.
	\end{proof}

The proof of Theorem \ref{th:gen_commute} shows that $S(\ket{\psi})$ 
is specified by only $\log_2 2^{n} = n$ {\em irredundant stabilizer generators}. 
Therefore, an arbitrary $n$-qubit stabilizer state can be represented by
a {\em stabilizer matrix} $\cM$ whose rows represent a set of 
generators $g_1,\ldots,g_n$ for $S(\ket{\psi})$. (Hence we use the terms 
{\em generator set} and {\em stabilizer matrix} interchangeably.) Since each $g_i$ is a string 
of $n$ Pauli literals, the size of the matrix is $n\times n$. The 
phases of each $g_i$ are stored separately using a vector of $n$ integers.
Therefore, the storage cost for $\cM$ is 
$\Theta(n^2)$, which is an {\em exponential 
improvement} over the $O(2^n)$ cost often encountered in 
vector-based representations.

Theorem \ref{th:gen_commute} suggests that Pauli literals can be 
represented using only two bits, e.g., $00 = I$, $01 = Z$, $10 = X$ and $11 = Y$. 
Therefore, a stabilizer matrix can be encoded using an $n\times2n$ binary matrix or {\em tableau}.
The advantage of this approach is that this literal-to-bits mapping induces an isomorphism
${\mathbb Z}_2^{2n} \rightarrow \mathcal{G}_n$ because vector addition
in ${\mathbb Z}_2^{2}$ is equivalent to multiplication of Pauli operators
up to a global phase. The tableau implementation of the stabilizer formalism
is covered in \cite{AaronGottes, NielChu}. 

    \begin{proposition} \label{prop:stab_count}
        The number of $n$-qubit pure stabilizer states is given by
        \begin{equation} \label{eq:stab_count}
    	   N(n) = 2^n\prod_{k=0}^{n-1}(2^{n-k} + 1) = 2^{(.5 + o(1))n^2}
    	   \vspace{-2pt}
	    \end{equation}
    \end{proposition}
    
The proof of Proposition \ref{prop:stab_count} can be found in~\cite{AaronGottes}. 
An alternate interpretation of Equation \ref{eq:stab_count} is
given by the simple recurrence relation $N(n)=2(2^n+1) N(n-1)$ with base case
$N(1) = 6$. For example, for $n=2$ the number
of stabilizer states is $60$, and for $n=3$ it is $1080$. This recurrence
relation stems from the fact that there are $2^n+1$ ways of combining the
generators of $N(n-1)$ with additional Pauli matrices to form valid
$n$-qubit generators. The factor of $2$ accounts for the increase in the number
of possible sign configurations. Table~\ref{tab:two_qbssts} and 
Appendix~\ref{app:three_qbssts} list all two-qubit and three-qubit 
stabilizer states, respectively.

	\begin{table}[!b]
		\centering
        \parbox{.75\linewidth}{\caption{\label{tab:two_qbssts} 
     Sixty two-qubit stabilizer states and their corresponding Pauli generators. Shorthand notation represents a stabilizer state as $\alpha_0, \alpha_1, \alpha_2, \alpha_3$ where $\alpha_i$ are the normalized amplitudes of the basis states. The basis states are emphasized in bold. The first column lists states whose generators do not include an upfront minus sign, and other columns introduce the signs. A sign change creates an orthogonal vector. Therefore, each row of the table gives an orthogonal basis. The cells in dark grey indicate stabilizer states with four non-zero basis amplitudes, i.e., $\alpha_i \neq 0\ \forall\ i$. The $\angle$ column indicates the angle between that state and $\ket{00}$, which has 12 nearest-neighbor states (light gray) and 15 orthogonal states ($\perp$).
     }}
        \scalebox{.90}[.90]{\begin{tabular}{|r|c|c|c||c|c|c||c|c|c||c|c|c|}
            \hline
            & \sc State & \hspace{-2.5mm} \sc Gen'tors \hspace{-2.5mm} & $\angle$
            & \sc State & \hspace{-2.5mm} \sc Gen'tors \hspace{-2.5mm} & $\angle$
            & \sc State & \hspace{-2.5mm} \sc Gen'tors \hspace{-2.5mm} & $\angle$
            & \sc State & \hspace{-2.5mm} \sc Gen'tors\hspace{-2.5mm}  & $\angle$ \\ \hline\hline
            \multirow{9}{2mm}{\rotatebox{90}{\sc Separable}}

            & \cellcolor[gray]{0.7} $1,1,1,1$  & IX, XI & \hspace{-3mm} $\pi/3$ \hspace{-3mm} & \cellcolor[gray]{0.7}  \hspace{-3mm} $1,-1,1,-1$  \hspace{-3mm} & -IX, XI  & \hspace{-3mm} $\pi/3$ \hspace{-3mm} & \cellcolor[gray]{0.7}  \hspace{-3mm} $1,1,-1,-1$  \hspace{-3mm} &  IX, -XI  & \hspace{-3mm} $\pi/3$ \hspace{-3mm} &  \cellcolor[gray]{0.7} $1,-1,-1,1$ &  -IX, -XI & \hspace{-3mm} $\pi/3$ \hspace{-3mm} \\ 
            & \cellcolor[gray]{0.7} $1,1,i,i$  &  IX, YI & \hspace{-3mm} $\pi/3$ \hspace{-3mm} & \cellcolor[gray]{0.7} $1,-1,i,-i$ & -IX, YI & \hspace{-3mm} $\pi/3$ \hspace{-3mm} & \cellcolor[gray]{0.7} $1,1,-i,-i$  &  IX, -YI & \hspace{-3mm} $\pi/3$ \hspace{-3mm} & \cellcolor[gray]{0.7} $1,-1,-i,i$  & -IX, -YI & \hspace{-3mm} $\pi/3$ \hspace{-3mm} \\ 
            & $1,1,0,0$ & IX, ZI & \cellcolor[gray]{0.85}\hspace{-3mm} $\pi/4$ \hspace{-3mm} & $1,-1,0,0$  & -IX, ZI & \cellcolor[gray]{0.85}\hspace{-3mm} $\pi/4$ \hspace{-3mm} & $0,0,1,1$  &  IX, -ZI & \hspace{-3mm} $\perp$ \hspace{-3mm} & $0,0,1,-1$ & -IX, -ZI & \hspace{-3mm} $\perp$ \hspace{-3mm} \\ 

            & \cellcolor[gray]{0.7} $1,i,1,i$  &  IY, XI & \hspace{-3mm} $\pi/3$ \hspace{-3mm} & \cellcolor[gray]{0.7} $1,-i,1,-i$ & -IY, XI & \hspace{-3mm} $\pi/3$ \hspace{-3mm} & \cellcolor[gray]{0.7} $1,i,-1,-i$ &  IY, -XI & \hspace{-3mm} $\pi/3$ \hspace{-3mm} & \cellcolor[gray]{0.7} $1,-i,-1,i$ & -IY, -XI & \hspace{-3mm} $\pi/3$ \hspace{-3mm} \\ 
            & \cellcolor[gray]{0.7} $1,i,i,-1$ &  IY, YI & \hspace{-3mm} $\pi/3$ \hspace{-3mm} & \cellcolor[gray]{0.7} $1,-i,i,1$  & -IY, YI & \hspace{-3mm} $\pi/3$ \hspace{-3mm} & \cellcolor[gray]{0.7} $1,i,-i,1$  &  IY, -YI & \hspace{-3mm} $\pi/3$ \hspace{-3mm} & \cellcolor[gray]{0.7} \hspace{-3mm} $1,-i,-i,-1$  \hspace{-3mm} & -IY, -YI & \hspace{-3mm} $\pi/3$ \hspace{-3mm}  \\ 
            & $1,i,0,0$ & IY, ZI & \cellcolor[gray]{0.85}\hspace{-3mm} $\pi/4$ \hspace{-3mm} & $1,-i,0,0$  & -IY, ZI & \cellcolor[gray]{0.85}\hspace{-3mm} $\pi/4$ \hspace{-3mm} & $0,0,1,i$ &  IY, -ZI & \hspace{-3mm} $\perp$ \hspace{-3mm} & $0,0,1,-i$ & -IY, -ZI & \hspace{-3mm} $\perp$ \hspace{-3mm} \\ 

            & $1,0,1,0$ &  IZ, XI & \cellcolor[gray]{0.85}\hspace{-3mm} $\pi/4$ \hspace{-3mm} & $0,1,0,1$  &  -IZ, XI & \hspace{-3mm} $\perp$ \hspace{-3mm} & $1,0,-1,0$  &  IZ, -XI & \cellcolor[gray]{0.85}\hspace{-3mm} $\pi/4$ \hspace{-3mm} & $0,1,0,-1$  & -IZ, -XI & \hspace{-3mm} $\perp$ \hspace{-3mm} \\ 
            & $1,0,i,0$ &  IZ, YI & \cellcolor[gray]{0.85}\hspace{-3mm} $\pi/4$ \hspace{-3mm} & $0,1,0,i$  &  -IZ, YI & \hspace{-3mm} $\perp$ \hspace{-3mm} & $1,0,-i,0$  &  IZ, -YI & \cellcolor[gray]{0.85}\hspace{-3mm} $\pi/4$ \hspace{-3mm} & $0,1,0,-i$  &  -IZ, -YI & \hspace{-3mm} $\perp$ \hspace{-3mm} \\ 
            & $\mathbf{1,0,0,0}$ & \bf IZ, ZI & $0$ & $\mathbf{0,1,0,0}$ & \bf -IZ, ZI  & \hspace{-3mm} $\perp$ \hspace{-3mm} & $\mathbf{0,0,1,0}$ & \bf IZ, -ZI  & \hspace{-3mm} $\perp$ \hspace{-3mm} &  $\mathbf{0,0,0,1}$  & \bf -IZ, -ZI & \hspace{-3mm} $\perp$ \hspace{-3mm} \\ \hline\hline

            \multirow{6}{2mm}{\rotatebox{90}{\sc Entangled}}

            & $0,1,1,0$ &  XX, YY  & \hspace{-3mm} $\perp$ \hspace{-3mm} &  $1,0,0,-1$ & -XX, YY & \cellcolor[gray]{0.85}\hspace{-3mm} $\pi/4$ \hspace{-3mm} & $1,0,0,1$  &  XX, -YY & \cellcolor[gray]{0.85}\hspace{-3mm} $\pi/4$ \hspace{-3mm} & $0,1,-1,0$ &  -XX, -YY & \hspace{-3mm} $\perp$ \hspace{-3mm} \\ 
            & $1,0,0,i$ &  XY, YX  & \cellcolor[gray]{0.85}\hspace{-3mm} $\pi/4$ \hspace{-3mm} &  $0,1,i,0$  & -XY, YX & \hspace{-3mm} $\perp$ \hspace{-3mm} & $0,1,-i,0$ &  XY, -YX & \hspace{-3mm} $\perp$ \hspace{-3mm} & $1,0,0,-i$ &  -XY, -YX & \cellcolor[gray]{0.85}\hspace{-3mm} $\pi/4$ \hspace{-3mm} \\ 
    		
    		& \cellcolor[gray]{0.7} \hspace{-3mm} $1,1,1,-1$ \hspace{-3mm} &  XZ, ZX  & \hspace{-3mm} $\pi/3$ \hspace{-3mm} &  \cellcolor[gray]{0.7} $1,1,-1,1$ & -XZ, ZX   & \hspace{-3mm} $\pi/3$ \hspace{-3mm} &  \cellcolor[gray]{0.7} $1,-1,1,1$  &  XZ, -ZX  & \hspace{-3mm} $\pi/3$ \hspace{-3mm} &  \cellcolor[gray]{0.7} \hspace{-3mm} $1,-1,-1,-1$ \hspace{-3mm} &  -XZ, -ZX & \hspace{-3mm} $\pi/3$ \hspace{-3mm} \\ 
            & \cellcolor[gray]{0.7} $1,i,1,-i$ &  XZ, ZY  & \hspace{-3mm} $\pi/3$ \hspace{-3mm} &  \cellcolor[gray]{0.7} $1,i,-1,i$ & -XZ, ZY & \hspace{-3mm} $\pi/3$ \hspace{-3mm} & \cellcolor[gray]{0.7} $1,-i,1,i$  &  XZ, -ZY  & \hspace{-3mm} $\pi/3$ \hspace{-3mm} &  \cellcolor[gray]{0.7} $1,-i,-1,-i$  &  -XZ, -ZY & \hspace{-3mm} $\pi/3$ \hspace{-3mm} \\ 

            & \cellcolor[gray]{0.7} $1,1,i,-i$ &  YZ, ZX  & \hspace{-3mm} $\pi/3$ \hspace{-3mm} &  \cellcolor[gray]{0.7} $1,1,-i,i$  &  -YZ, ZX  & \hspace{-3mm} $\pi/3$ \hspace{-3mm} &  \cellcolor[gray]{0.7} $1,-1,i,i$ & YZ, -ZX  & \hspace{-3mm} $\pi/3$ \hspace{-3mm} & \cellcolor[gray]{0.7} $1,-1,-i,-i$  &  -YZ, -ZX & \hspace{-3mm} $\pi/3$ \hspace{-3mm} \\ 
            & \cellcolor[gray]{0.7} $1,i,i,1$  &  YZ, ZY  & \hspace{-3mm} $\pi/3$ \hspace{-3mm} &  \cellcolor[gray]{0.7} $1,i,-i,-1$ &  -YZ, ZY  & \hspace{-3mm} $\pi/3$ \hspace{-3mm} &  \cellcolor[gray]{0.7} $1,-i,i,-1$  & YZ, -ZY  & \hspace{-3mm} $\pi/3$ \hspace{-3mm} &  \cellcolor[gray]{0.7} $1,-i,-i,1$   &  -YZ, -ZY & \hspace{-3mm} $\pi/3$ \hspace{-3mm} \\
            \hline
        \end{tabular}
        }
    \end{table} 
    
    \begin{observation} \label{obs:stabst_amps}
    	Consider a stabilizer state $\ket{\psi}$ represented by a set of generators of its
    	stabilizer group $S(\ket{\psi})$. Recall from the proof of Theorem \ref{th:gen_commute}
    	that, since $S(\ket{\psi}) \cong {\mathbb Z}_2^n$,
    	each generator imposes a linear constraint on $\ket{\psi}$. Therefore, the set
        of generators can be viewed as a system of linear equations whose solution
        yields the $2^n$ basis amplitudes that make up $\ket{\psi}$. Thus,
        one needs to perform Gaussian elimination to obtain the basis
        amplitudes from a generator set. 
    \end{observation} 
	
\noindent
{\bf Canonical stabilizer matrices}.
Although stabilizer states are uniquely determined by their stabilizer group, the set of 
generators may be selected in different ways. For example, the state 
$\ket{\psi} = (\ket{00} + \ket{11})/\sqrt{2}$ is uniquely specified by any of 
the following stabilizer matrices:

	\begin{center}
		\begin{tabular}{c|l|cc|l|cc|l|}
		\multirow{2}{8mm}{$\cM_1 =$} & $XX$ & & \multirow{2}{8mm}{$\cM_2 =$} & $XX$ & &
		\multirow{2}{8mm}{$\cM_3 =$} & -$YY$ \\
		& $ZZ$ & & & -$YY$ & & & $ZZ$ \\
		\end{tabular}
	\end{center} 
	
	\begin{figure}[!t]
	\centering
	\begin{tabular}{cc}
	\includegraphics[scale=.30]{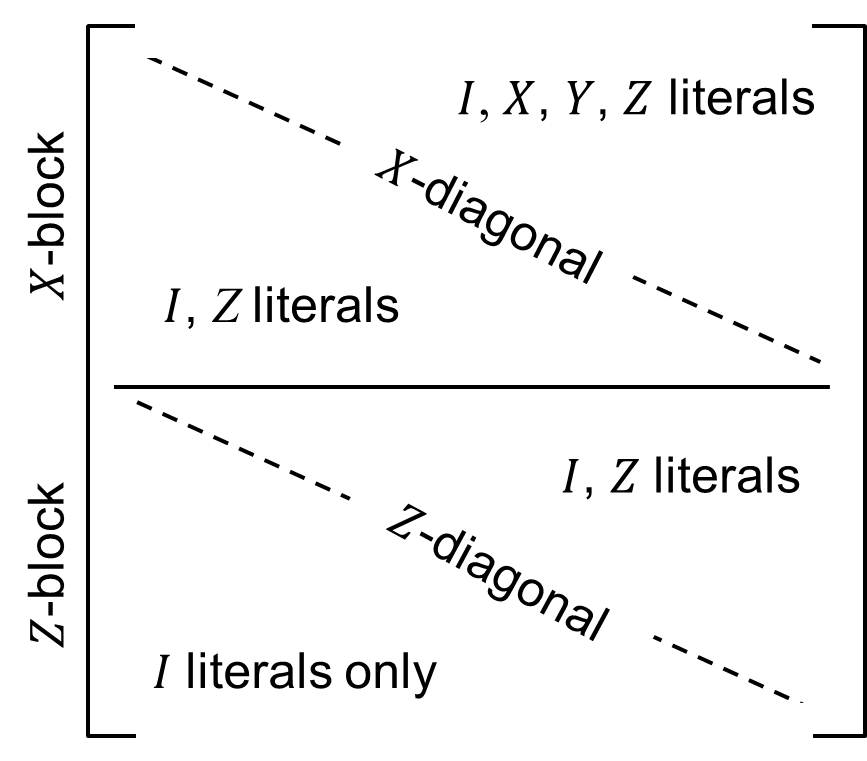} 
	&
	\raisebox{55pt}{\parbox{.45\linewidth}{
	\caption{\label{fig:sminv} 
	Canonical (row-reduced echelon) form for stabilizer matrices. 
	The $X$-block contains a {\em minimal} set of rows with $X$/$Y$ literals. 
	The rows with $Z$ literals only appear in the $Z$-block. Each block 
	is arranged so that the leading non-$I$ literal of each row is strictly 
	to the right of the leading non-$I$ literal in the row above. The number
	of Pauli (non-$I$) literals in each block is minimal.}}
	}
	\end{tabular}
	\vspace{-10pt}
	\end{figure}

\noindent
One obtains $\cM_2$ from $\cM_1$ by left-multiplying the second row
by the first. Similarly, one can also obtain $\cM_3$ from $\cM_1$ 
or $\cM_2$ via row multiplication. Observe that, multiplying any row by itself yields
$II$, which stabilizes $\ket{\psi}$. However, $II$ cannot
be used as a stabilizer generator because it is redundant
and carries no information about the structure of $\ket{\psi}$. 
This also holds true in general for $\cM$ of any size.
Any stabilizer matrix can be rearranged
by applying sequences of elementary row operations in order to obtain a particular
matrix structure. Such operations do not modify the stabilizer state.
The elementary row operations that can be performed on a stabilizer matrix are
transposition, which swaps two rows of the matrix, and multiplication,
which left-multiplies one row with another. Such operations allow one
to rearrange the stabilizer matrix in a series of steps that resemble
Gauss-Jordan elimination.\footnote{\scriptsize Since Gaussian elimination
essentially inverts the $n\times 2n$ matrix, this
could be sped up to $O(n^{2.376})$ time by using
fast matrix inversion algorithms. However, $O(n^3)$-
time Gaussian elimination seems more practical.} ~Given an $n\times n$ stabilizer matrix, 
row transpositions are performed in constant time\footnote{\scriptsize Storing pointers to rows
facilitates $O(1)$-time row transpositions -- one simply swaps 
relevant pointers.} ~while row multiplications require $\Theta(n)$ time.
Algorithm~\ref{alg:gauss_min} rearranges a stabilizer matrix into
a {\em row-reduced echelon form} that contains: ({\em i}) a {\em minimum 
set} of generators with $X$ and $Y$ literals appearing at the top, 
and ({\em ii}) generators containing a {\em minimum set} of $Z$ literals 
only appearing at the bottom of the matrix.
This particular stabilizer-matrix structure, shown in Figure \ref{fig:sminv},  
defines a canonical representation for stabilizer states \cite{Djor, Gottes98}. 
The algorithm iteratively determines which 
row operations to apply based on the Pauli (non-$I$) literals
contained in the first row and column of an increasingly smaller submatrix
of the full stabilizer matrix. Initially, the submatrix considered is the full stabilizer matrix.
After the proper row operations are applied, the dimensions of the submatrix
decrease by one until the size of the submatrix reaches one. 
The algorithm performs this process twice, once to position the
rows with $X$($Y$) literals at the top, and then again to position
the remaining rows containing $Z$ literals only at the bottom. 
Let $i\in\{1,\ldots,n\}$ and $j\in\{1,\ldots,n\}$ be the index
of the first row and first column, respectively, of submatrix
$\cA$. The steps to construct the upper-triangular portion of the row-echelon
form shown in Figure \ref{fig:sminv} are as follows.  

	\begin{itemize}
		\item[{\bf 1.}] Let $k$ be a row in $\cA$ whose $j^{th}$ literal
		is $X$($Y$). Swap rows $k$ and $i$ such that $k$ is the first row of $\cA$.
		Decrease the height of $\cA$ by one (i.e., increase $i$).
		\item[{\bf 2.}] For each row $m \in \{0,\ldots,n\}, m \neq i$ that has
		an $X$($Y$) in column $j$, use row multiplication to set the $j^{th}$
		literal in row $m$ to $I$ or $Z$.
		\item[{\bf 3.}] Decrease the width of $\cA$ by one (i.e., increase $j$).
	\end{itemize}
	
\newcommand{\rswap}{\mathrm{\tt ROWSWAP}}
\newcommand{\rmult}{\mathrm{\tt ROWMULT}}

    \begin{algorithm}[!b]
        \caption{Canonical form reduction for stabilizer matrices}
        \label{alg:gauss_min}
        \footnotesize
        \begin{algorithmic}[1]
            \Require Stabilizer matrix $\cM$ for $S(\ket{\psi})$ with rows $R_1,\ldots,R_n$
            \Ensure $\cM$ is reduced to row-echelon form
            \Statex \hspace{-5mm} $\Rightarrow$ $\rswap(\cM, i, j)$ swaps rows $R_i$ and $R_j$ of $\cM$
            \Statex \hspace{-5mm} $\Rightarrow$ $\rmult(\cM, i, j)$ left-multiplies rows $R_i$ 
            and $R_j$, returns updated $R_i$
            \Statex
            \State $i \leftarrow 1$
            \For{$j \in \{1, \dots, n\}$} \Comment{Setup $X$ block}
            	\State $k \leftarrow$ index of row $R_{k \in \{i,\ldots, n\}}$ 
            	with $j^{th}$ literal set to $X$($Y$)
            	\If{$k$ {\bf exists}}
            		\State $\rswap(\cM, i, k)$
            		\For{$m \in \{0,\ldots, n\}$} 
            			\If{$j^{th}$ literal of $R_m$ is $X$ or $Y$ and $m\neq i$}
            				\State $R_m = \rmult(\cM, R_i, R_m)$ \Comment{Gauss-Jordan elimination step}
            			\EndIf
            		\EndFor
            		\State $i \leftarrow i + 1$
            	\EndIf
            \EndFor
           	\For{$j \in \{1, \dots, n\}$} \Comment{Setup $Z$ block}
            	\State $k \leftarrow$ index of row $R_{k \in \{i,\ldots, n\}}$ 
            	with $j^{th}$ literal set to $Z$
            	\If{$k$ {\bf exists}}
            		\State $\rswap(\cM, i, k)$
            		\For{$m \in \{0,\ldots, n\}$}
            			\If{$j^{th}$ literal of $R_m$ is $Z$ or $Y$ and $m\neq i$}
            				\State $R_m = \rmult(\cM, R_i, R_m)$ \Comment{Gauss-Jordan elimination step}
            			\EndIf
            		\EndFor
            		\State $i \leftarrow i + 1$
           	\EndIf
            \EndFor
        \end{algorithmic}
    \end{algorithm}

To bring the matrix to its lower-triangular form, one executes the same process
with the following difference: ($i$) step~1 looks for rows that have
a $Z$ literal (instead of $X$ or $Y$) in column $j$, and ($ii$) step~2
looks for rows that have $Z$ or $Y$ literals (instead of $X$ or $Y$) in column $j$.
Observe that Algorithm~\ref{alg:gauss_min} ensures that the columns
in $\cM$ have at most two distinct types of non-$I$ literals.
Since Algorithm~\ref{alg:gauss_min} 
inspects all $n^2$ entries in the matrix and performs a constant number of row
multiplications each time, the runtime of the algorithm is $O(n^3)$. 
An alternative row-echelon form for stabilizer generators along 
with relevant algorithms to obtain them were introduced in \cite{Audenaert}.
However, their matrix structure is not canonical as it does not guarantee
a minimum set of generators with $X$/$Y$ literals.

\ \\ \noindent
{\bf Stabilizer circuit simulation}.
The computational basis states are stabilizer states that can be represented using the 
following stabilizer-matrix structure.

	\begin{definition} \label{def:basis_form}
		A stabilizer matrix is in {\em \normform} if it has the following structure.
		\begin{equation*}
			\begin{array}{r}
				\pm \\
				\pm \\
				\vdots \\
				\pm
			\end{array}
			\left[
			\begin{array}{cccc}
				Z & I & \cdots & I  \\
				I & Z & \cdots & I   \\
				\vdots & \vdots & \ddots & \vdots \\
				I & I & \cdots & Z
			\end{array}
 			\right]
		\end{equation*}
	\end{definition}
	
In this matrix form, the $\pm$ sign of each row along with 
its corresponding $Z_j$-literal designates whether 
the state of the $j^{th}$ qubit is $\ket{0}$ ($+$) or $\ket{1}$ ($-$). Suppose we 
want to simulate circuit $\cC$. Stabilizer-based
simulation first initializes $\cM$ to specify some basis 
state $\ket{\psi}$. To simulate the action of each gate $U \in \cC$,
we conjugate each row $g_i$ of $\cM$ by $U$.\footnote{\scriptsize
Since $g_i\ket{\psi} = \ket{\psi}$, 
the resulting state $U\ket{\psi}$ is stabilized by $Ug_iU^\dag$ 
because $(Ug_iU^\dag) U\ket{\psi} = Ug_i\ket{\psi} = U\ket{\psi}$.}
~We require that $Ug_iU^\dag$ maps to another
string of Pauli literals so that the resulting stabilizer
matrix $\cM'$ is well-formed. It turns out that the Hadamard,
Phase and CNOT gates (Figure \ref{fig:chp_pauli}a) have such mappings, 
i.e., these gates conjugate the Pauli group onto itself \cite{Gottes98, NielChu}.
Table~\ref{tab:cliff_mult} lists the mapping for each of these
gates. 

For example, suppose we simulate a CNOT operation
on $\ket{\psi} = (\ket{00} + \ket{11})/\sqrt{2}$ using $\cM$,
\begin{center}
	\begin{tabular}{c|l|cc|l|}
		\multirow{2}{6mm}{$\cM =$}   & $XX$ & \multirow{2}{8mm}{$\xrightarrow{CNOT}$} 
		&\multirow{2}{7mm}{$\cM' =$} & $XI$ \\
		& $ZZ$ &   						&                         &  $IZ$ \\
	\end{tabular}
\end{center}

\noindent
One can verify that the rows of $\cM'$ stabilize
$\ket{\psi}\xrightarrow{CNOT}(\ket{00} + \ket{10})/\sqrt{2}$ as required.

Since Hadamard, Phase and CNOT gates are directly simulated
using stabilizers, these gates are commonly 
called \emph{stabilizer gates}. They are 
also called \emph{Clifford gates} because they generate the Clifford group of unitary
operators. We use these names
interchangeably. Any circuit composed exclusively of stabilizer
gates is called a \emph{unitary stabilizer circuit}. 
Table~\ref{tab:cliff_mult} shows that at most two columns of $\cM$ 
are updated when one simulates a stabilizer gate. Thus, such gates are simulated in $\Theta(n)$
time. 

    \begin{theorem} \label{th:stabst}
        An $n$-qubit stabilizer state $\ket{\psi}$ can be obtained
        by applying a stabilizer circuit to the $\ket{0}^{\otimes n}$
        basis state. 
    \end{theorem}
    \begin{proof}
    	The work in \cite{AaronGottes} represents the generators using a tableau, 
    	and then shows how to construct a unitary stabilizer circuit from the 
    	tableau. We refer the reader to \cite[Theorem 8]{AaronGottes} for 
    	details of the proof.
    \end{proof}

    \begin{corollary} \label{cor:stab_allzeros}
	   	An $n$-qubit stabilizer state $\ket{\psi}$ can be transformed
  		by stabilizer gates into the $\ket{0}^{\otimes n}$ basis state.
   	\end{corollary}
   	\begin{proof}
        Since every stabilizer state can be produced by applying some unitary
        stabilizer circuit $\cC$ to the $\ket{0}^{\otimes n}$ state, it suffices to reverse
        $\cC$ to perform the inverse transformation. To reverse a stabilizer
        circuit, reverse the order of gates and replace every $P$ gate with $PPP$.
   \end{proof}
   
   \begin{table}[!t]
 		\parbox{.55\linewidth}{\caption{\label{tab:cliff_mult} Conjugation of the Pauli-group elements
    	by the stabilizer gates \cite{NielChu}. For the CNOT case, subscript $1$
    	indicates the control and $2$ the target.}	}
        \centering \footnotesize
        \begin{tabular}{cc}
        \begin{tabular}{|c||c|c|}
            \hline
                \sc Gate & \sc Input  & \sc Output \\ \hline\hline
                          & $X$ & $Z$ \\
                $H$       & $Y$ & -$Y$ \\
                          & $Z$ & $X$  \\ \hline
                          & $X$ & $Y$ \\
                $P$       & $Y$ & -$X$ \\
                          & $Z$ & $Z$   \\
            \hline
        \end{tabular}
        & \hspace{-5pt}
         \begin{tabular}{|c||c|c|}
            \hline
            \sc Gate & \sc Input  & \sc Output  \\ \hline\hline
            \multirow{6}{10mm}{$CNOT$} & $I_1X_2$ & $I_1X_2$   \\
                      & $X_1I_2$ & $X_1X_2$    \\
                      & $I_1Y_2$ & $Z_1Y_2$      \\
                      & $Y_1I_2$ & $Y_1X_2$     \\
                      & $I_1Z_2$ & $Z_1Z_2$   \\
                      & $Z_1I_2$ & $Z_1I_2$    \\
            \hline
        \end{tabular} \\ \\
		\end{tabular}
		\vspace{-10pt}
 	\end{table}
	
The stabilizer formalism also admits one-qubit measurements 
in the computational basis \cite{Gottes98}. 
However, the updates to $\cM$ for such gates are not as efficient
as for stabilizer gates. Note that any qubit
$j$ in a stabilizer state is either in a $\ket{0}$ ($\ket{1}$) state or 
in an unbiased\footnote{\scriptsize 
An arbitrary state $\ket{\psi}$ with computational basis decomposition
$\sum_{k=0}^n \lambda_k\ket{k}$ is said to be {\em unbiased}
if for all $\lambda_i \neq 0$ and $\lambda_j\neq 0$,
$|\lambda_i|^2 = |\lambda_j|^2$. Otherwise,
the state is {\em biased}. One can verify that none of the stabilizer 
gates produce biased states.}~superposition 
of both. The former case is called a {\em deterministic
outcome} and the latter a {\em random outcome}. We can tell these 
cases apart in $\Theta(n)$ time by searching for $X$ or $Y$ literals in
the $j^{th}$ column of $\cM$. If such literals are found, the qubit must be in 
a superposition and the outcome is random with equal probability
($p(0) = p(1) = .5$); otherwise the outcome is deterministic
($p(0) = 1$ or $p(1) = 1$). 

	{\em Random case}: one flips an unbiased coin to decide the outcome 
	and then updates $\cM$ to make it consistent with the 
    outcome obtained. This requires at most $n$ row multiplications 
    leading to $O(n^2)$ runtime \cite{AaronGottes, NielChu}. 
	
	{\em Deterministic case}: no updates to 
	$\cM$ are necessary but we need to figure out whether the state 
	of the qubit is $\ket{0}$ or $\ket{1}$, i.e., whether the qubit is stabilized by $Z$ or -$Z$,
	respectively. One approach is to apply Algorithm \ref{alg:gauss_min} to put $\cM$ in  
	row-echelon form. 
	This removes redundant literals from $\cM$ in order
	to identify the row containing a $Z$ in its $j^{th}$ position
	and $I$ everywhere else. The $\pm$ phase of this row 
	decides the outcome of the measurement. Since this approach
	is a form of Gaussian elimination, it takes $O(n^3)$ time in 
	practice.   

Aaronson and Gottesman \cite{AaronGottes} improved the runtime of deterministic
measurements by doubling the size of $\cM$ to include $n$ {\em destabilizer generators} 
in addition to the $n$ stabilizer generators. Such destabilizer generators help 
identify exactly which row multiplications to compute in order 
to decide the measurement outcome. This approach avoids Gaussian elimination
and thus deterministic measurements are computed in $O(n^2)$ time.

\section{Inner-product and circuit-synthesis algorithms}  
\label{sec:inprod_stab}

Given $\braket{\psi}{\varphi} = re^{i\alpha}$, we
normalize the global phase of $\ket{\psi}$ to ensure, without loss
of generality, that $\braket{\psi}{\varphi} \in \mathbb{R}_+$.

	\begin{theorem} \label{th:stab_ortho}
		Let $S(\ket{\psi})$ and $S(\ket{\varphi})$ be the stabilizer
		groups for $\ket{\psi}$ and $\ket{\varphi}$, respectively.
		If there exist $P \in S(\ket{\psi})$ and $Q \in S(\ket{\varphi})$
		such that $P =$ -$Q$, then $\ket{\psi}\perp\ket{\varphi}$.
	\end{theorem}
	\begin{proof}
		Since $\ket{\psi}$ is a $1$-eigenvector of $P$ and
		$\ket{\varphi}$ is a $(-1)$-eigenvector of $P$, they
		must be orthogonal.
	\end{proof}

	\begin{theorem}\label{th:inprod_aron}{\em\bf \cite{AaronGottes}}
		Let $\ket{\psi}$, $\ket{\varphi}$ be non-orthogonal
		stabilizer states. Let $s$ be the minimum, over all sets of
		generators $\{P_1,\ldots,P_n\}$ for $S(\ket{\psi})$ and
		$\{Q_1,\ldots,Q_n\}$ for $S(\ket{\varphi})$, of the number
		of different $i$ values for which $P_i \neq Q_i$. Then,
		$|\braket{\psi}{\varphi}|=2^{-s/2}$.
	\end{theorem}
	\begin{proof}
		Since $\braket{\psi}{\varphi}$ is not affected by
		unitary transformations $U$,
		we choose a stabilizer circuit
		such that $U\ket{\psi} = \ket{b}$, where $\ket{b}$ is a basis state.
		For this state, select the stabilizer generators $\cM$ of the
		form $I\ldots IZI\ldots I$. Perform Gaussian elimination
		on $\cM$ to minimize the incidence of $P_i \neq Q_i$.
		Consider two cases. If $U\ket{\varphi} \neq \ket{b}$ and its generators
		contain only $I$/$Z$ literals, then $U\ket{\varphi}\perp U\ket{\psi}$,
		which contradicts the assumption that $\ket{\psi}$ and $\ket{\varphi}$
		are non-orthogonal. Otherwise, each generator of
		$U\ket{\varphi}$ containing $X$/$Y$ literals contributes 
		a factor of $1/\sqrt{2}$ to the inner product.
	\end{proof}
	
\ \\ \noindent
{\bf Synthesizing canonical circuits}.
A crucial step in the proof of Theorem \ref{th:inprod_aron} is the computation of a stabilizer
circuit that brings an $n$-qubit stabilizer state $\ket{\psi}$ to a computational basis state $\ket{b}$.
Consider a stabilizer matrix $\cM$ that uniquely identifies $\ket{\psi}$. $\cM$ is reduced 
to \normform\ (Definition~\ref{def:basis_form}) by applying a series of elementary row and column 
operations. Recall that row operations
(transposition and multiplication) do not modify the state, but column (Clifford) operations do.
Thus, the column operations involved in the reduction process constitute a unitary stabilizer
circuit $\cC$ such that $\cC\ket{\psi} = \ket{b}$, where $\ket{b}$ is a basis state. 
Algorithm~\ref{alg:inprod_circ} reduces an input stabilizer matrix $\cM$ to \normform\ 
and returns the circuit $\cC$ that performs such a mapping. 

	\begin{definition}
		Given a finite sequence of quantum gates,
		a {\em circuit template} describes a segmentation
		of the circuit into blocks where each block
		uses only one gate type.
		The blocks must correspond to the sequence and be
  		concatenated in that order. For example,
  		a circuit satisfying the $H$-$C$-$P$ template starts with
  		a block of Hadamard ($H$) gates, followed by a block of
  		CNOT ($C$) gates, followed by a block of Phase ($P$) gates.
	\end{definition}
	
	\begin{definition}
		A circuit with a {\em template structure} consisting
		entirely of CNOT, Hadamard and Phase blocks is called
		a {\em canonical stabilizer circuit}.
	\end{definition}
	
\newcommand{\zfield}{\mathbb{Z}}

Canonical forms are useful for synthesizing stabilizer circuits
that minimize the number of gates and qubits required to produce a
particular computation. This is particularly important in the context
of quantum fault-tolerant architectures that are based on stabilizer codes.
Given any stabilizer matrix, Algorithm~\ref{alg:inprod_circ} synthesizes a $5$-block
canonical circuit with template $H$-$C$-$CZ$-$P$-$H$ (Figure~\ref{fig:basiscirc}-a),
where the $CZ$ block consists of Controlled-$Z$ (CPHASE) gates. Such
gates are stabilizer gates since CPHASE$_{i,j} = $H$_j$CNOT$_{i,j}$H$_j$ 
(Figure~\ref{fig:basiscirc}-b). In our implementation, such gates are simulated 
directly on the stabilizer. The work in \cite{AaronGottes} establishes 
a longer $7$-block\footnote{Theorem 8 in \cite{AaronGottes}
actually describes an $11$-step canonical procedure. However, the last four steps
pertain to reducing destabilizer rows, which we do not consider in our approach.} 
$H$-$C$-$P$-$C$-$P$-$C$-$H$ canonical-circuit template. The existence of a $H$-$C$-$P$-$CZ$ template
is proven in \cite{Vanden} but no algorithms are known 
for obtaining such $4$-block canonical circuits given an arbitrary 
stabilizer state.

	\newcommand{\apply}{\mathrm{\tt CONJ}}
	\newcommand{\gauss}{\mathrm{\tt GAUSS}}
	
    \begin{algorithm}[!t]
         \caption{Synthesis of basis normalization circuit}
         \small
         \label{alg:inprod_circ}
         \begin{algorithmic}[1]
            \Require Stabilizer matrix $\cM$ for $S(\ket{\psi})$ with rows $R_1,\ldots,R_n$
            \Ensure ({\em i}) Unitary stabilizer circuit $\cC$ such that $\cC\ket{\psi}$
            equals basis state $\ket{b}$, and ({\em ii}) reduce $\cM$ to \normform 
            \Statex \hspace{-5mm} $\Rightarrow$ $\gauss(\cM)$ reduces $\cM$ to canonical form 
            (Figure \ref{fig:sminv})
            \Statex \hspace{-5mm} $\Rightarrow$ $\rswap(\cM, i, j)$ swaps rows $R_i$ and $R_j$ of $\cM$
            \Statex \hspace{-5mm} $\Rightarrow$ $\rmult(\cM, i, j)$ left-multiplies rows $R_i$ 
            and $R_j$, returns updated $R_i$
            \Statex \hspace{-5mm} $\Rightarrow$ $\apply(\cM, \alpha_j)$ conjugates $j^{th}$ 
            column of $\cM$ by Clifford sequence $\alpha$
            \Statex
            \State $\gauss(\cM)$ \Comment{Set $\cM$ to canonical form}
            \State $\cC \leftarrow \emptyset$   
            \State $i \leftarrow 1$
            \For{$j \in \{1, \dots, n\}$} 
            \Comment{Apply block of Hadamard gates}
            	\State $k \leftarrow$ index of row $R_{k \in \{i,\ldots, n\}}$ with $j^{th}$ literal set to $X$ or $Y$
            	\If{$k$ {\bf exists}}
            		\State $\rswap(\cM, i, k)$
            	\Else
            		\State $k_2 \leftarrow$ index of {\em last} row 
            		$R_{k_2 \in \{i,\ldots, n\}}$ with $j^{th}$ literal set to $Z$
            		\If{$k_2$ {\bf exists}}
            			\State $\rswap(\cM, i, k_2)$
            			\If{$R_{i}$ has $X$, $Y$ or $Z$ literals in columns $\{j+1, \ldots, n\}$}
            				\State $\apply(\cM, \text{H}_j)$
            				\State $\cC \leftarrow \cC \cup \text{H}_j$
            			\EndIf
           		\EndIf
           	\EndIf
           	\State $i \leftarrow i + 1$
            \EndFor
            \For{$j \in \{1, \dots, n\}$} \Comment{Apply block of CNOT gates}
            		 \For{$k \in \{j+1, \dots, n\}$} 
            		 	\If{$k^{th}$ literal of row $R_j$ is set to $X$ or $Y$}
            				\State $\apply(\cM, \text{CNOT}_{j, k})$
            				\State $\cC \leftarrow \cC \cup \text{CNOT}_{j, k}$
            			\EndIf
           		 \EndFor
            \EndFor
            \For{$j \in \{1, \dots, n\}$} 
            \Comment{Apply a block of Controlled-$Z$ gates (Figure~\ref{fig:basiscirc}b)}
            		\For{$k \in \{j+1, \dots, n\}$} 
            			\If{$k^{th}$ literal of row $R_j$ is set to $Z$}
            				\State $\apply(\cM, \text{CPHASE}_{j, k})$ 
            				\State $\cC \leftarrow \cC \cup \text{CPHASE}_{j, k}$
            			\EndIf
           		\EndFor
            \EndFor
            \For{$j \in \{1, \dots, n\}$} \Comment{Apply block of Phase gates}
            		\If{$j^{th}$ literal of row $R_j$ is set to $Y$}
            			\State $\apply(\cM, \text{P}_j)$
            			\State $\cC \leftarrow \cC \cup \text{P}_j$
            		\EndIf
            \EndFor
            \For{$j \in \{1, \dots, n\}$} \Comment{Apply block of Hadamard gates}
            		\If{$j^{th}$ literal of row $R_j$ is set to $X$}
            			\State $\apply(\cM, \text{H}_j)$
            			\State $\cC \leftarrow \cC \cup \text{H}_j$
            		\EndIf
            \EndFor
            \For{$j \in \{1, \dots, n\}$} 
            \Comment{Eliminate trailing $Z$ literals to ensure \normform\ (Definition~\ref{def:basis_form})}
            		\For{$k \in \{j+1, \dots, n\}$} 
            			\If{$j^{th}$ literal of row $R_k$ is set to $Z$}
            				\State $R_k = \rmult(\cM, R_j, R_k)$ 
            			\EndIf
           		\EndFor
            \EndFor
            \State \Return $\cC$
        \end{algorithmic}
    \end{algorithm}
    
We now describe the main steps in Algorithm~\ref{alg:inprod_circ}.
For simplicity, the updates to the phase array under row and column operations will 
be left out of our discussion as such updates do not affect the overall execution 
of the algorithm.

	\begin{figure}[!t]
		\centering
		\begin{tabular}{cc}
		\scalebox{.8}[.8]{
			\begin{minipage}[b]{.4\linewidth}
				\input{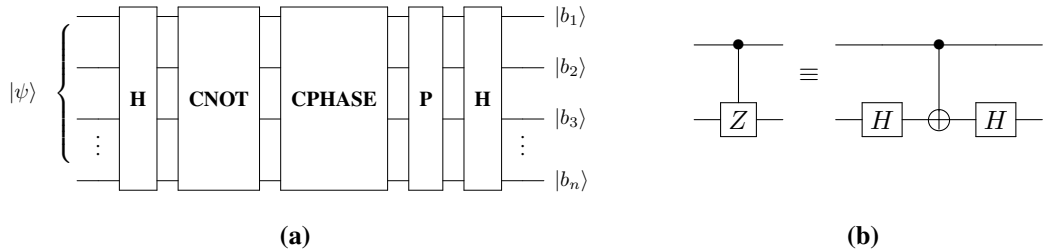}
			\end{minipage}
		} 
		& \hspace{50pt} 
		$
    		\Qcircuit @C=1.0em @R=1.0em {
    			& \\
    			& \ctrl{2} & \qw & & & \qw 	   & \ctrl{2} & \qw & \qw \\
    			& & & \equiv & & & \\
			& \gate{Z} & \qw & & & \gate{H}  & \targ    & \gate{H}  & \qw \\
   	 	}$
   	 	\\ \\
		{\bf (a)} & \hspace{50pt} {\bf (b) }
		\end{tabular}
		\parbox{.75\linewidth}{
		\caption{\label{fig:basiscirc}{\bf(a)} Template structure for the basis-normalization 
		circuit synthesized by Algorithm~\ref{alg:inprod_circ}. The input is an arbitrary
		stabilizer state $\ket{\psi}$ while the output is a basis state $\ket{b_1, \ldots, b_n}$,
		where $b_1, \ldots, b_n \in \{0,1\}^n$. {\bf (b)} Controlled-$Z$
		gates used in the {\bf CPHASE} block. CPHASE gates can be implemented directly or using
		the equivalence shown here. 		
		}}
		\vspace{-10pt}
    \end{figure}

	\begin{itemize}
		\item[{\bf 1.}] Reduce $\cM$ to canonical form.
		\item[{\bf 2.}] Use row transposition to diagonalize $\cM$. For $j \in \{1, \ldots, n\}$, 
		if the diagonal literal $\cM_{j, j} = Z$ and there are other Pauli (non-$I$) 
		literals in the row (qubit is entangled), conjugate $\cM$ by H$_j$. 
		Elements below the diagonal are $Z$/$I$ literals. 
		\item[{\bf 3.}] For each above-diagonal element $\cM_{j, k} = X$/$Y$, 
		conjugate by CNOT$_{j,k}$. Elements above the diagonal are now $I$/$Z$ literals.
		\item[{\bf 4.}] For each above-diagonal element $\cM_{j, k} = Z$, 
		conjugate by CPHASE$_{j,k}$. Elements above the diagonal are now $I$ literals.
		\item[{\bf 5.}] For each diagonal literal $\cM_{j, j} = Y$, conjugate by P$_j$.
		\item[{\bf 6.}] For each diagonal literal $\cM_{j, j} = X$, conjugate by H$_j$.
		\item[{\bf 7.}] Use row multiplication to eliminate trailing $Z$ literals below
		the diagonal and arrive at \normform.
	\end{itemize}
    
	\begin{proposition} \label{prop:normform_circsize}
		For an $n\times n$ stabilizer matrix $\cM$, the number of 
		gates in the circuit $\cC$ returned by Algorithm~\ref{alg:inprod_circ} 
		is $O(n^2)$.
	\end{proposition}
	\begin{proof}
		The number of gates in $\cC$ is dominated by the CNOT and CPHASE blocks, 
		which have $O(n^2)$ gates each. This agrees with previous
		results regarding the number of gates needed for an $n$-qubit 
		stabilizer circuit in the worst case \cite{Cleve, DehaeDemoor}.
	\end{proof} 

Observe that, for each gate added to $\cC$, the corresponding
column operation is applied to $\cM$. Therefore, since column 
operations run in $\Theta(n)$ time, it follows from Proposition~\ref{prop:normform_circsize}
that the runtime of Algorithm \ref{alg:inprod_circ} is $O(n^3)$. 

Canonical stabilizer circuits that follow the $7$-block template structure from \cite{AaronGottes}
can be optimized to obtain a tighter bound on the number of gates. As in 
our approach, such circuits are dominated by the size of the CNOT blocks,
which contain $O(n^2)$ gates. The work in \cite{PatelMarkov} shows that 
that any CNOT circuit has an equivalent CNOT circuit with only $O(n^2/\log n)$ gates.
Thus, one simply applies such techniques to each of the CNOT blocks in the
canonical circuit. It is an open problem whether one can apply the techniques 
from \cite{PatelMarkov} directly to CPHASE blocks, which would 
facilitate similar optimizations to our proposed $5$-block canonical form.  

\newcommand{\norm}{\mathrm{\tt BASISNORMCIRC}}
\newcommand{\lmult}{\mathrm{\tt LEFTMULT}}
	
        \begin{algorithm}[!t]
            \caption{Inner product for stabilizer states}
            \label{alg:inprod} \small
            \begin{algorithmic}[1]
                \Require Stabilizer matrices ({\em i}) $\cM^\psi$ for $\ket{\psi}$ with rows $P_1,\ldots,P_n$,
                and ({\em ii}) $\cM^\phi$ for $\ket{\phi}$ with rows $Q_1,\ldots,Q_n$
                \Ensure Inner product between $\ket{\psi}$ and $\ket{\phi}$
                \Statex \hspace{-5mm} $\Rightarrow$ $\norm(\cM)$ reduces $\cM$ to \normform, 
                i.e, $\cC\ket{\psi}=\ket{b}$, where $\ket{b}$ is a basis state, and returns $\cC$
                \Statex \hspace{-5mm} $\Rightarrow$  $\apply(\cM, \cC)$ conjugates $\cM$ by Clifford circuit $\cC$
                \Statex \hspace{-5mm} $\Rightarrow$ $\gauss(\cM)$ reduces $\cM$ to canonical form 
                (Figure \ref{fig:sminv})
                \Statex \hspace{-5mm} $\Rightarrow$ $\lmult(P, Q)$ left-multiplies Pauli operators $P$ and $Q$, 
                and returns the updated $Q$
                \Statex
                \State $\cC \leftarrow \norm(\cM^\psi)$
                \Comment{Apply Algorithm \ref{alg:inprod_circ} to $\cM^\psi$}  
                	\State $\apply(\cM^\phi, \cC)$  \Comment{Compute $\cC\ket{\phi}$}
                \State $\gauss(\cM^\phi)$ \Comment{Set $\cM^\phi$ to canonical form}
                \State $k \leftarrow 0$
                \For{{\bf each} row $Q_i \in \cM^\phi$}
                	\If{$Q_i$ contains $X$ or $Y$ literals}
	              		\State $k \leftarrow k + 1$
	              	\Else\Comment{Check orthogonality, i.e., $Q_i \notin S(\ket{b})$.}
	              		\State $R \leftarrow I^{\otimes n}$
	              		\For{{\bf each} $Z$ literal in $Q_i$ found at position $j$}
	              			\State $R \leftarrow \lmult(P_j, R)$
		              	\EndFor
		              	\If{$R = -Q_i$}            	
		              		\State\Return 0 \Comment{By Theorem \ref{th:stab_ortho}}		  
		              	\EndIf
	              	\EndIf
                \EndFor
                \State\Return $2^{-k/2}$ \Comment{By Theorem \ref{th:inprod_aron}}
            \end{algorithmic}
        \end{algorithm}
	
\ \\ \noindent
{\bf Inner-product algorithm}.
Let $\ket{\psi}$ and $\ket{\phi}$ be two stabilizer states represented by stabilizer matrices
$\cM^\psi$ and $\cM^\phi$, respectively. Our approach for computing the inner product between these two
states is shown in Algorithm~\ref{alg:inprod}. Following the proof of Theorem \ref{th:inprod_aron},
Algorithm \ref{alg:inprod_circ} is applied to $\cM^\psi$ in order to reduce it to \normform.
The stabilizer circuit generated by Algorithm~\ref{alg:inprod_circ} is then applied
to $\cM^\phi$ in order to preserve the inner product. Then, we minimize the number of $X$ and $Y$ literals
in $\cM^\phi$ by applying Algorithm~\ref{alg:gauss_min}. Finally,
each generator in $\cM^\phi$ that anticommutes with $\cM^\psi$ (since $\cM^\psi$ is in \normform, we only
need to check which generators in $\cM^\phi$ have $X$ or $Y$ literals) contributes a factor of $1/\sqrt{2}$ 
to the inner product. If a generator in $\cM^\phi$, say $Q_i$, commutes with $\cM^\psi$, then we
check orthogonality by determining whether $Q_i$ is in the stabilizer group generated
by $\cM^\psi$. This is accomplished by multiplying the appropriate generators in $\cM^\psi$ such that
we create Pauli operator $R$, which has the same literals as $Q_i$,
and check whether $R$ has an opposite sign to $Q_i$. If this is the case, then, by
Theorem \ref{th:stab_ortho}, the states are orthogonal. 
Clearly, the most time-consuming step of Algorithm \ref{alg:inprod} is
the call to Algorithm \ref{alg:inprod_circ}, therefore, the overall runtime is $O(n^3)$.
However, as we show in Section \ref{sec:validate}, the performance of 
our algorithm depends strongly on the stabilizer matrices considered
and exhibits quadratic behaviour for certain stabilizer states. 

\section{Empirical Validation} \label{sec:validate}

We implemented our algorithms in C++ and designed a benchmark set 
to validate the performance of our inner-product algorithm. Recall 
that the runtime of Algorithm~\ref{alg:inprod_circ} is dominated by
the two nested for-loops (lines 20-35). 
The number of times these loops execute depends on the amount of 
entanglement in the input stabilizer state. In turn, the number of entangled 
qubits depends on the the number of CNOT gates in the 
circuit $\cC$ used to generate the stabilizer state $\cC\ket{0^{\otimes n}}$
(Theorem \ref{th:stabst}).
By a simple heuristic argument \cite{AaronGottes},
one generates highly entangled stabilizer states as long as the 
number of CNOT gates in $\cC$ is proportional to $n\lg n$. Therefore, we generated 
random $n$-qubit stabilizer circuits for $n\in\{20, 40, \ldots, 500\}$ as follows: 
fix a parameter $\beta > 0$; then choose $\beta \lceil n \log_2 n\rceil$ 
unitary gates (CNOT, Phase or Hadamard) each with probability $1/3$.
Then, each random $\cC$ is applied to the $\ket{00\ldots0}$ basis state to generate
random stabilizer matrices (states). The use of randomly generated benchmarks is 
justified for our experiments because (\emph{i}) our algorithms are not explicitly 
sensitive to circuit topology and (\emph{ii}) random stabilizer circuits are 
considered representative \cite{Knill}.
For each $n$, we applied Algorithm \ref{alg:inprod} to pairs of random
stabilizer matrices and measured the number of seconds needed to
compute the inner product. The entire procedure was repeated for
increasing degrees of entanglement by ranging $\beta$
from $0.6$ to $1.2$ in increments of $0.1$. Our results
are shown in Figure~\ref{fig:stabip}-a.

\begin{figure}[!b]
	\centering
	\begin{tabular}{cc} 
		\includegraphics[scale=.55]{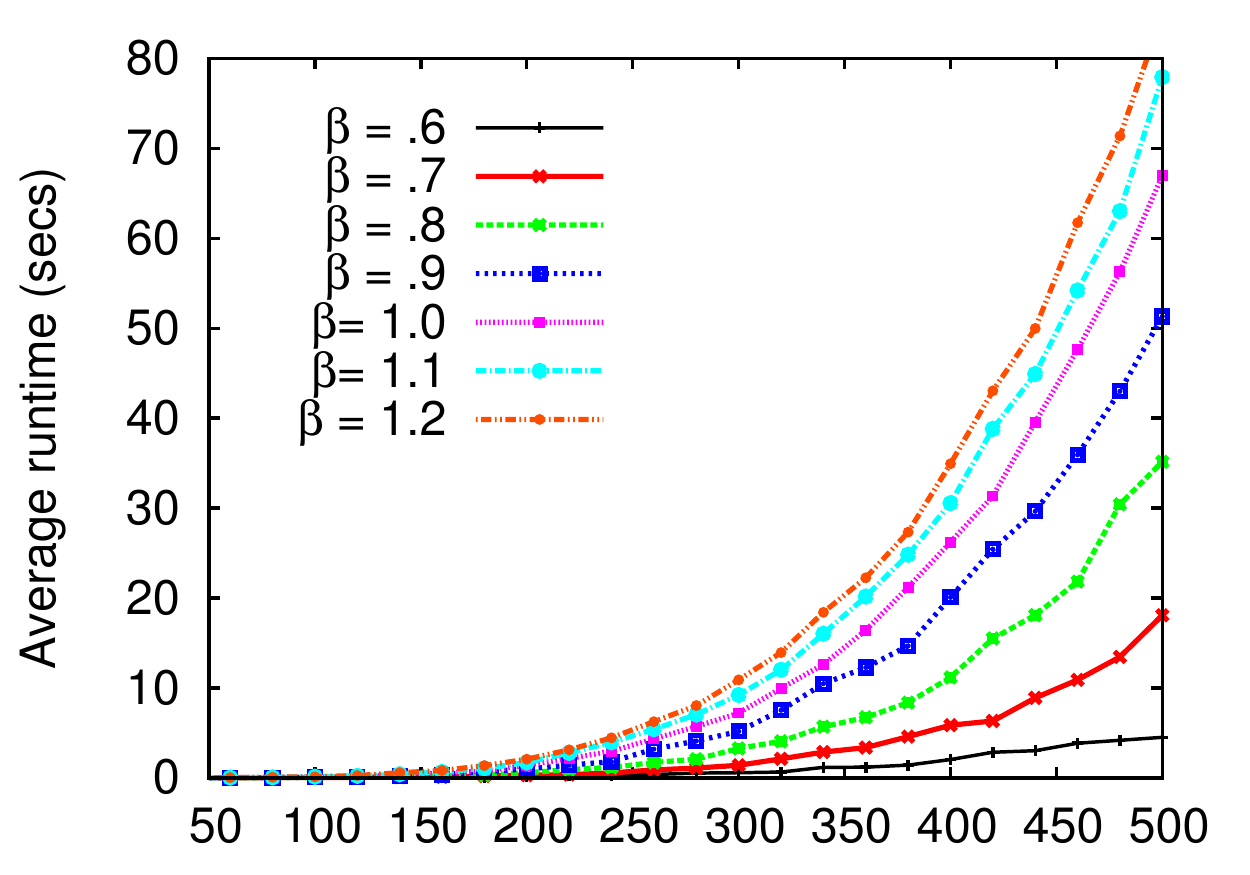} 
		&
		\includegraphics[scale=.55]{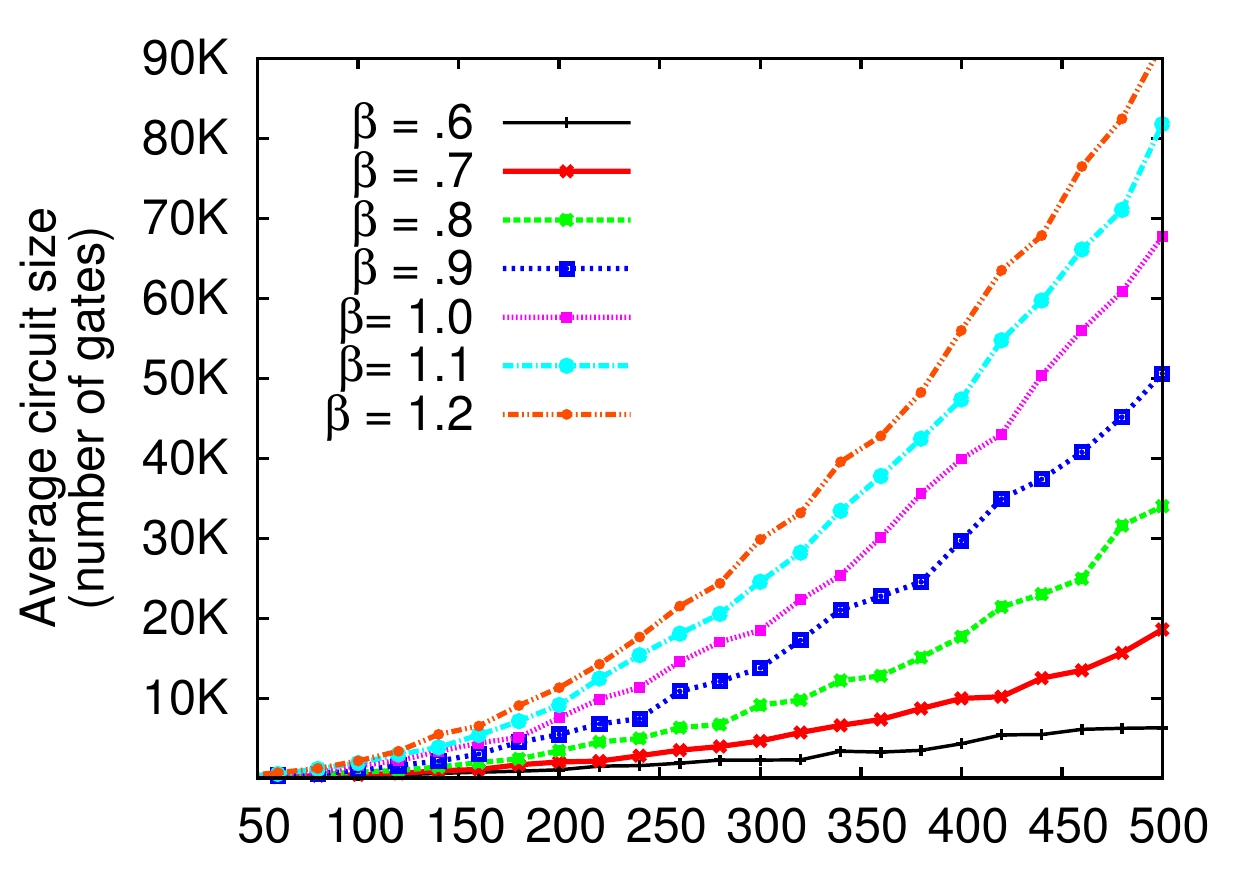} \\
		\multicolumn{2}{c}{Number of qubits} \\
		\hspace{1cm} {\bf (a)} & {\bf (b)} 
	\end{tabular}
	\parbox{.75\linewidth}{
	\caption{\label{fig:stabip} Average runtime for Algorithm \ref{alg:inprod}
	to compute the inner product between two random $n$-qubit stabilizer states. The stabilizer
	matrices that represent the input states are generated by applying $\beta n\log_2 n$
	unitary gates to $\ket{0^{\otimes n}}$. 
	}}
\end{figure}

The runtime of Algorithm~\ref{alg:inprod} 
appears to grow quadratically in $n$ when $\beta = 0.6$. However, 
when the number of unitary gates is doubled
($\beta = 1.2$), the runtime exhibits cubic growth.
Therefore, Figure~\ref{fig:stabip}-a shows that 
the performance of Algorithm~\ref{alg:inprod} is highly 
dependent on the degree of entanglement in the input
stabilizer states. Figure~\ref{fig:stabip}-b shows the
average size of the basis-normalization circuit returned
by the calls to Algorithm~\ref{alg:inprod_circ}. As expected
(Proposition~\ref{prop:normform_circsize}), the size of the
circuit grows quadratically in $n$.
Figure~\ref{fig:stabip_ghz0all} shows the average runtime for 
Algorithm \ref{alg:inprod} to compute the inner 
product between: (\emph{i}) the all-zeros basis state        
and random $n$-qubit stabilizer
states, and (\emph{ii}) the $n$-qubit GHZ
state\footnote{An $n$-qubit 
{\em GHZ state} is an equal superposition of the all-zeros and 
all-ones states, i.e., $\frac{\ket{0^{\otimes n}} + \ket{1^{\otimes n}}}{\sqrt{2}}$.} 
and random stabilizer states. GHZ states are maximally entangled
states that have been realized experimentally using several quantum technologies
and are often encountered in practical applications such as error-correcting 
codes and fault-tolerant architectures. Figure \ref{fig:stabip_ghz0all}
shows that, for such practical instances, Algorithm \ref{alg:inprod} 
can compute the inner product in roughly $O(n^2)$ time (e.g. $\braket{GHZ}{0}$). 
However, without apriori information about the input stabilizer matrices, 
one can only say that the performance of Algorithm \ref{alg:inprod} will 
be somewhere between quadratic and cubic in $n$.

\section{Nearest-neighbor stabilizer states} \label{sec:stabneighbors}

\begin{figure}[!t]
	\centering
	\begin{tabular}{cc}
		\includegraphics[scale=.55]{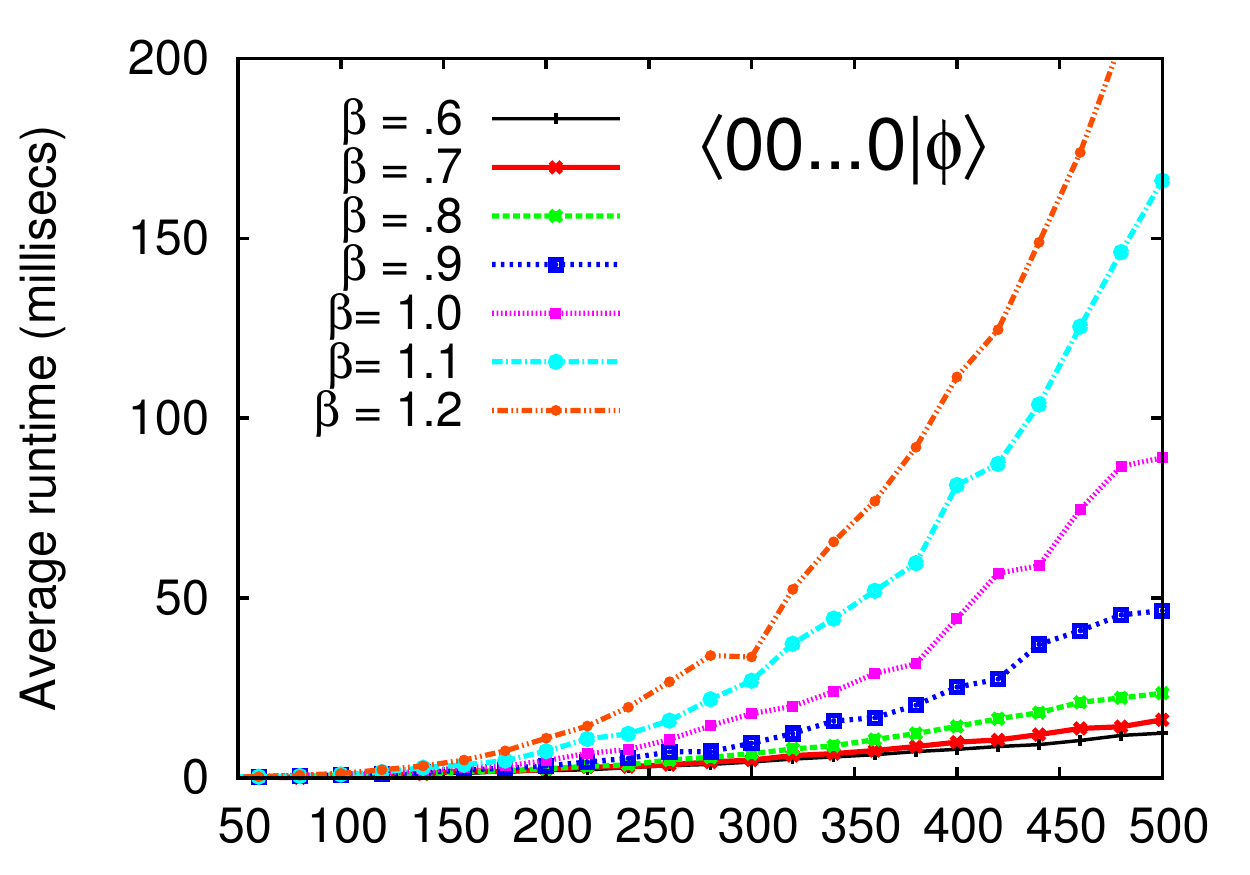} &
		\includegraphics[scale=.55]{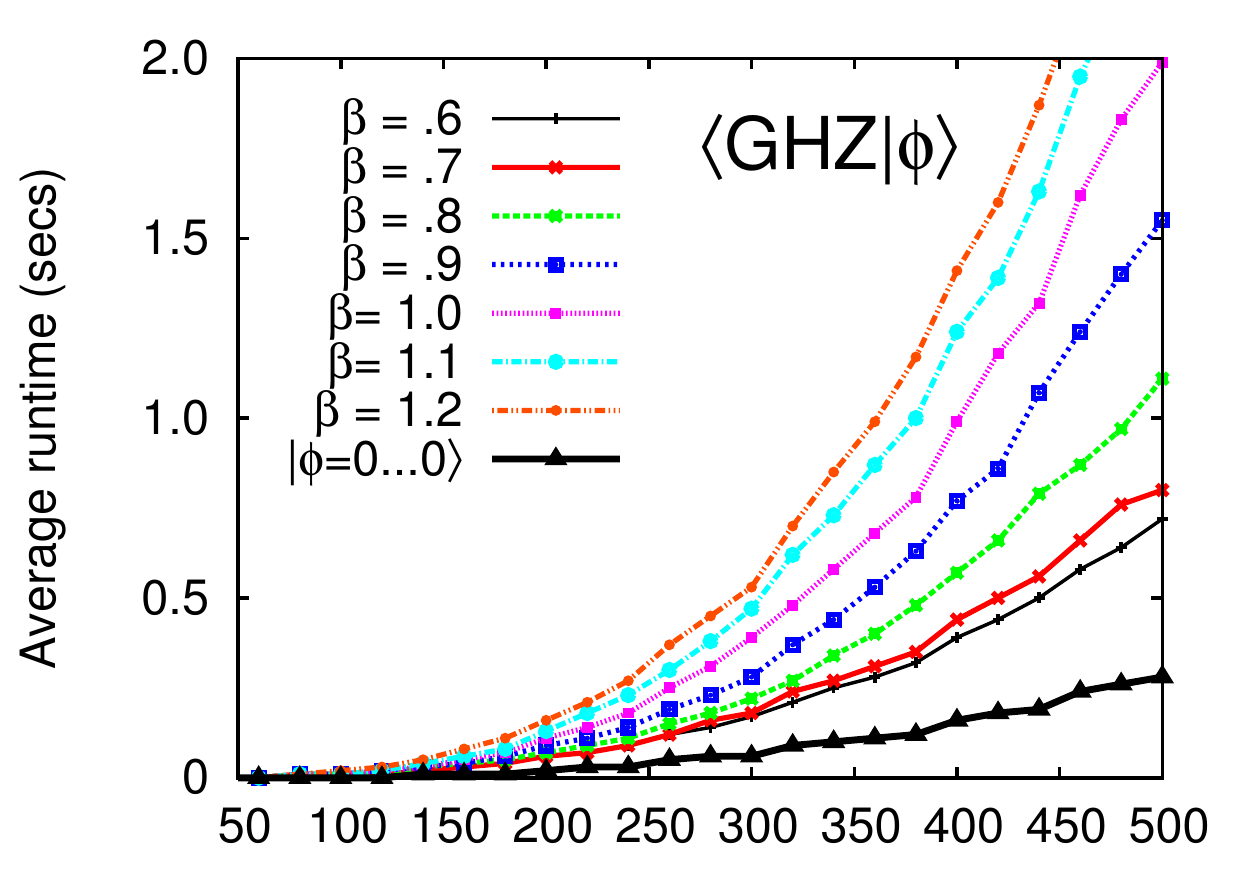} \\
		\multicolumn{2}{c}{Number of qubits} \\
		{\bf (a)} & {\bf (b)} 
	\end{tabular}
	\parbox{.75\linewidth}{
	\caption{\label{fig:stabip_ghz0all} Average runtime for Algorithm \ref{alg:inprod}
	to compute the inner product between {\bf (a)} $\ket{0^{\otimes n}}$ and random 
	stabilizer state $\ket{\phi}$ and {\bf (b)} the $n$-qubit GHZ state and random 
	stabilizer state $\ket{\phi}$.}
	}
	\vspace{-10pt}
\end{figure}

We used Algorithm~\ref{alg:inprod} to compute the inner product between 
$\ket{00}$ and all two-qubit stabilizer states. Our results
are shown in Table~\ref{tab:two_qbssts}. We leveraged these 
results to formulate the following properties related to the 
geometry of stabilizer states.

    \begin{definition}
		Given an arbitrary state $\ket{\psi}$ with $||\psi|| = 1$, a stabilizer state
		$\ket{\varphi}$ is a nearest stabilizer state
		to $\ket{\psi}$ if $|\braket{\psi}{\varphi}|$ attains the largest possible value $\neq 1$.
	\end{definition}

	\begin{proposition}
		Consider two orthogonal stabilizer states $\ket{\alpha}$ and $\ket{\beta}$
		whose unbiased superposition $\ket{\psi}$
		is also a stabilizer state. Then $\ket{\psi}$ is a
		nearest stabilizer state to $\ket{\alpha}$ and $\ket{\beta}$.
	\end{proposition}
	\begin{proof}
		Since stabilizer states are unbiased, $|\braket{\psi}{\alpha}| =
		|\braket{\psi}{\beta}| = \frac{1}{\sqrt{2}}$. By Theorem \ref{th:inprod_aron},
		this is the largest possible value $\neq 1$. Thus, $\ket{\psi}$
		is a nearest stabilizer state to $\ket{\alpha}$ and $\ket{\beta}$.
	\end{proof}
	
	\begin{lemma} \label{lem:num_stabsts}
		For any two stabilizer states, the numbers of nearest-neighbor
        stabilizer states are equal.
  	\end{lemma}
  	\begin{proof}
		By Corollary \ref{cor:stab_allzeros}, 
		any stabilizer state can be mapped to another stabilizer state by a
		stabilizer circuit. Since the operators effected by these circuits 
		are unitary, inner products are preserved.
	\end{proof}
	
	\begin{lemma}\label{lem:cross}
        Let $|\psi\rangle$ and $|\varphi\rangle$ be orthogonal stabilizer states such
        that $|\varphi\rangle=P|\psi\rangle$ where $P$ is an element of the Pauli group.
        Then $\frac{|\psi\rangle+|\varphi\rangle}{\sqrt{2}}$ is a stabilizer state.
    \end{lemma}
    \begin{proof}
        Suppose $|\psi\rangle=\langle g_k\rangle_{k=1,2,\dots,n}$
        is generated by elements $g_k$ of the $n$-qubit Pauli group.
        Let 
        	\[
        		f(k)=\left\{\begin{array}{cl} 0 & \textrm{if}\ [P,g_k]=0 \\
        									  1 & \textrm{otherwise}\end{array}\right.
      			\vspace{-4pt}
        	\]
        
        \noindent
        and write $|\varphi\rangle=\langle (-1)^{f(k)}g_k\rangle$.
        Conjugating each generator $g_k$ by $P$ we see that $\ket{\varphi}$
        is stabilized by $\langle (-1)^{f(k)} g_k\rangle$.
        Let $Z_k$ (respectively $X_k$) denote the Pauli operator $Z$ ($X$) acting
        on the $k^{th}$ qubit. By Corollary \ref{cor:stab_allzeros}, there exists
        an element $L$ of the $n$-qubit Clifford group such that $L|\psi\rangle=\ket{0}^{\otimes n}$
        and $L|\varphi\rangle=(LPL^\dag)L|\psi\rangle=i^t|f(1)f(2)\dots f(n)\rangle$.
        The second equality follows from the fact that $LPL^\dag$ is an element of
        the Pauli group and can therefore be written as $i^tX(v)Z(u)$ for some
        $t \in \{0,1,2,3\}$ and $u,v \in {\mathbb Z}_2^k$. Therefore,
        	\vspace{-6pt}
        	\begin{equation*}
		        \frac{|\psi\rangle+|\varphi\rangle}{\sqrt{2}}= \frac{L^\dag(\ket{0}^{\otimes n}
        		+ i^t\ket{f(1)f(2)\ldots f(n)})}{\sqrt{2}}
        	\end{equation*}
        The state in parenthesis on the right-hand side is the product of an
        all-zeros state and a GHZ state. Therefore, the sum is stabilized by 
        $S' = L^\dag\langle S_{zero}, S_{ghz}\rangle L$
        where $S_{zero} = \langle Z_i, i \in \{k|f(k)=0\}\rangle$ and $S_{ghz}$
        is supported on $\{k|f(k)=1\}$ and equals $\langle(-1)^{t/2}XX\ldots X,\forall i\ Z_iZ_{i+1}\rangle$
        if $t = 0\mod 2$ or $\langle(-1)^{(t-1)/2}YY\ldots Y,\forall i\ Z_iZ_{i+1}\rangle$
        if $t = 1\mod 2$.
    \end{proof}
	
	\begin{theorem}
        For any $n$-qubit stabilizer state $\ket{\psi}$, there are $4(2^n - 1)$
        nearest-neighbor stabilizer states, and these states can be produced as
        described in Lemma \ref{lem:cross}.
	\end{theorem}
    \begin{proof}
		The all-zeros basis amplitude of any stabilizer state $\ket{\psi}$
		that is a nearest neighbor to $\ket{0}^{\otimes n}$
    		must be $\propto 1/\sqrt{2}$. Therefore, $\ket{\psi}$ is an unbiased superposition of
        $\ket{0}^{\otimes n}$ and one of the other $2^n-1$ basis states,
        i.e., $\ket{\psi} = \frac{\ket{0}^{\otimes n} + P\ket{0}^{\otimes n}}{\sqrt{2}}$,
        where $P \in \mathcal{G}_n$ such that 
        $P\ket{0}^{\otimes n} \neq \alpha\ket{0}^{\otimes n}$.
        As in the proof of Lemma~\ref{lem:cross},
        we have $\ket{\psi} = \frac{\ket{0}^{\otimes n} + i^t\ket{\varphi}}{\sqrt{2}}$,
        where $\ket{\varphi}$ is a basis state and $t \in \{0,1,2,3\}$. Thus,
        there are $4$ possible unbiased superpositions, and a total of
        $4(2^n - 1)$ nearest stabilizer states. Since $\ket{0}^{\otimes n}$ is a stabilizer
        state, all stabilizer states have the same number of nearest stabilizer states
        by Lemma~\ref{lem:num_stabsts}.
    \end{proof}
    
Table~\ref{tab:two_qbssts} shows that $\ket{00}$ has $12$ nearest-neighbor 
states. We computed inner products between all-pairs of $2$-qubit stabilizer
states and confirmed that each had exactly $12$ nearest neighbors. We used
the same procedure to verify that all $3$-qubit stabilizer states
have $28$ nearest neighbors. We verified the correctness of our algorithm
by comparing against inner product computations based on explicit basis
amplitudes.

\section{Stabilizer frames} \label{sec:stabframes}

Given an $n$-qubit stabilizer state $\ket{\psi}$, there exists an
orthonormal basis including $\ket{\psi}$ and consisting entirely of 
stabilizer states. Using Theorem \ref{th:stab_ortho}, one can generate
such a basis from the stabilizer representation of $\ket{\psi}$. 
Observe that, one can create a
state $\ket{\varphi}$ that is orthogonal to $\ket{\psi}$ by
changing the signs of an arbitrary non-empty subset of generators
of $S(\ket{\psi})$, i.e., by permuting the phase vector of the
stabilizer matrix for $\ket{\psi}$.
Moreover, selecting two different subsets will produce two mutually
orthogonal states. Thus, one can produce $2^n-1$ additional orthogonal stabilizer states.
Such states, together with $\ket{\psi}$, form an orthonormal basis. This is illustrated 
by Table~\ref{tab:two_qbssts} were each row constitutes an orthonormal basis.

	\begin{definition} \label{def:stab_frame}
		A {\em stabilizer frame} $\cF$ is a set of $k\leq 2^n$ stabilizer states
		that forms an orthonormal basis $\{\ket{\psi_1},\dots, \ket{\psi_k}\}$ 
		and spans a subspace of the $n$-qubit Hilbert space.
		$\cF$ is represented by a pair consisting of 
		({\em i}) a stabilizer matrix $\cM$ and ({\em ii}) a set of $k$ 
		distinct phase vectors $\sigma_j(\cM), j \in \{1,\ldots, k\}$. 
		The size of the frame, which we denote by $|\cF|$, is equal to $k$. 
	\end{definition}
	
\renewcommand\vec[1]{\ensuremath\boldsymbol{#1}}
	
Stabilizer frames are useful for representing arbitrary quantum states and
for simulating the action of stabilizer circuits on such states. Let 
$\vec{\alpha} = (\alpha_1, \ldots, \alpha_k) \in \mathbb{C}^k$ be 
the decomposition of the arbitrary $n$-qubit state $\ket{\phi}$ onto the basis 
$\{\ket{\psi_1},\dots, \ket{\psi_k}\}$ defined by $\cF$, i.e., 
$\ket{\phi} = \sum_{i=1}^k\alpha_k\ket{\psi_i}$. Furthermore, let
$U$ be a stabilizer gate. To simulate $U\ket{\phi}$, 
one simply {\em rotates the basis defined by $\cF$}
to get the new basis $\{U\ket{\psi_1},\dots, U\ket{\psi_k}\}$. This is accomplished
with the following two-step process: ({\em i}) update the stabilizer matrix 
$\cM$ associated with $\cF$ as per Section~\ref{sec:stab}; ({\em ii})  
iterate over the phase vectors in $\cF$ and update each accordingly 
(Table~\ref{tab:cliff_mult}). The second step is linear in the number 
of phase vectors as only a constant number of elements in each 
vector needs to be updated. Also, $\vec{\alpha}$ may
need to be updated, which requires the computation of the global 
phase of each $U\ket{\psi_i}$. Since the stabilizer
does not maintain global phases directly, each $\alpha_i$ is updated as follows:

	\begin{itemize}
		\item[{\bf 1.}] Use Gaussian elimination to obtain a basis state 
		$\ket{b}$ from $\cM$ (Observation~\ref{obs:stabst_amps}) and store 
		its non-zero amplitude $\beta$. If $U$ is the Hadamard gate, it
		may be necessary to sample a sum of two non-zero (one real, one 
		imaginary) basis amplitudes. 
		\item[{\bf 2.}] Compute $U\beta\ket{b}=\beta'\ket{b'}$ directly 
		using the state-vector representation.
		\item[{\bf 3.}] Obtain $\ket{b'}$ from $U\cM U^\dag$ and 
		store its non-zero amplitude $\gamma$. 
		\item[{\bf 4.}] Compute the global-phase factor generated 
		as $\alpha_i=(\alpha_i\cdot\beta')/\gamma$.
	\end{itemize}


Observe that, all the above processes take time polynomial in $k$, therefore,
if $k = poly(n)$, $U\ket{\phi}$ can be simulated {\em efficiently} on 
a classical computer via frame-based simulation. 

\ \\ \noindent
{\bf Inner product between frames}. We now discuss how to use our algorithms
to compute the inner product between arbitrary quantum states. 
Let $\ket{\phi}$ and $\ket{\varphi}$ be quantum states 
represented by the pairs $<\cF^\phi,\ \vec{\alpha}= (\alpha_1, \ldots, \alpha_k)>$ and 
$<\cF^\varphi,\ \vec{\beta}= (\beta_1, \ldots, \beta_l)>$, respectively. 
The following steps compute $|\braket{\phi}{\varphi}|$.

	\begin{itemize}
		\item[{\bf 1.}] Apply Algorithm~\ref{alg:inprod_circ} to $\cM^\phi$ 
		(the stabilizer matrix associated with $\cF^\phi$)
		to obtain basis-normalization circuit $\cC$. 
		\item[{\bf 2.}] Rotate frames $\cF^\phi$ and $\cF^\varphi$ by $\cC$ as outlined
		in our previous discussion.
		\item[{\bf 3.}] Reduce $\cM^\phi$ to canonical form 
		(Algorithm~\ref{alg:gauss_min}) and record the row operations applied. 
		Apply the same row operations to each phase vector $\sigma^\phi_i, {i\in \{1,\ldots,k\}}$
		in $\cF^\phi$. Repeat this step for $\cM^\varphi$ and the phase vectors
		in $\cF^\varphi$. 
		\item[{\bf 4.}] Let $\cM^\phi_i$ denote that the leading-phases of the rows
		in $\cM^\phi$ are set equal to $\sigma^\phi_i$. Similarly, $\cM^\varphi_j$ denotes
		that the phases of $\cM^\varphi$ are equal to $\sigma^\varphi_j$.
		Furthermore, let $\delta(\cM^\phi_i, \cM^\varphi_j)$ be the
		function that returns $0$ if the orthogonality check from Algorithm~\ref{alg:inprod} 
		(lines 9--15) returns $0$, and $1$ otherwise. The inner product is computed as,
		\vspace{-5pt}
		\begin{equation*}
			|\braket{\phi}{\varphi}| = \frac{1}{2^{s/2}}\sum_{i=1}^k\sum_{j=1}^l
			|\alpha_i^*\beta_j|\cdot\delta(\cM^i_\phi, \cM^j_\varphi)
		\end{equation*}
		\noindent
		where $s$ is the number of rows in $\cM_\varphi$ that contain $X$ or $Y$ literals. 
	\end{itemize}
	
Prior work on representation
of arbitrary states using the stabilizer formalism can be found in \cite{AaronGottes}.
The authors propose an approach that represents a quantum state as a sum 
of density-matrix terms. Our frame-based technique offers more compact storage 
($|\cF| \leq 2^n$ whereas a density matrix may have $4^n$ non-zero entries)
but requires more sophisticated book-keeping.
	
\section{Conclusion} \label{sec:conclude}

The stabilizer formalism facilitates compact representation of 
stabilizer states and efficient simulation of stabilizer circuits.
Stabilizer states appear in many different quantum-information applications,
and their efficient manipulation via geometric and linear-algebraic operations
may lead to additional insights. To this end, we study algorithms to 
efficiently compute the inner product between stabilizer states. 
A crucial step of this computation is the synthesis
of a canonical circuit that transforms a stabilizer state into a computational
basis state. We designed an algorithm to synthesize such circuits
using a $5$-block template structure and showed that these circuits 
contain $O(n^2)$ stabilizer gates.
We analysed the performance of our inner-product algorithm and 
showed that, although its runtime is $O(n^3)$, there are practical instances 
in which it runs in linear or quadratic time. Furthermore, we proved that 
an $n$-qubit stabilizer state has exactly $4(2^n-1)$ nearest-neighbor
states and verified this result experimentally. Finally, we designed 
techniques for representing arbitrary quantum states
using stabilizer frames and generalize our algorithms to compute 
the inner product between two such frames.



\bibliographystyle{IEEEtran}
%

\newpage

\appendices
 
\section{The 1080 three-qubit stabilizer states}
\label{app:three_qbssts}

Shorthand notation represents a stabilizer state as $\alpha_0, \alpha_1, \alpha_2, \alpha_3$ where $\alpha_i$ are the normalized amplitudes of the basis states. Basis states are emphasized in bold. The $\angle$ column indicates the angle between that state and $\ket{000}$, which has $28$ nearest-neighbor states and $315$ orthogonal states ($\perp$). \index{stabilizer state!three-qubit}

\begin{center}
	\vspace{10pt}
	\scriptsize
	\scalebox{.6}[.5]{
\begin{tabular}{|c|c||c|c||c|c||c|c||c|c||c|c||c|c||c|c|}\hline
\sc State & \sc Gen'tors  & \sc State & \sc Gen'tors  & \sc State & \sc Gen'tors  & \sc State & \sc Gen'tors  & \sc State & \sc Gen'tors  & \sc State & \sc Gen'tors  & \sc State & \sc Gen'tors  & \sc State & \sc Gen'tors \\ \hline\hline
\hspace{-2mm}$1 1 1 1 1 1 1 1$\hspace{-2mm} & \hspace{-2mm}IIX IXI XII\hspace{-2mm} & \hspace{-2mm}$1 1 1 1 $-$ $-$ $-$ $-$$\hspace{-2mm} & \hspace{-2mm}IIX IXI -XII\hspace{-2mm} & \hspace{-2mm}$1 0 0 i $-$ 0 0 \bar{i} $\hspace{-2mm} & \hspace{-2mm}IXY IYX -XII\hspace{-2mm} & \hspace{-2mm}$0 \bar{i}  0 0 1 0 0 0$\hspace{-2mm} & \hspace{-2mm}IZI YIX -XIY\hspace{-2mm} & \hspace{-2mm}$1 1 $-$ $-$ $-$ $-$ 1 1$\hspace{-2mm} & \hspace{-2mm}IIX -IXI -XII\hspace{-2mm} & \hspace{-2mm}$1 \bar{i}  i $-$ $-$ i \bar{i}  1$\hspace{-2mm} & \hspace{-2mm}IYZ -IXX -XII\hspace{-2mm} & \hspace{-2mm}$1 0 $-$ 0 0 i 0 i$\hspace{-2mm} & \hspace{-2mm}XYX -IXZ -YYY\hspace{-2mm} & \hspace{-2mm}$1 $-$ $-$ 1 $-$ 1 1 $-$$\hspace{-2mm} & \hspace{-2mm}-IIX -IXI -XII\hspace{-2mm}\hspace{-2mm}\\ 
\hspace{-2mm}$1 1 1 1 i i i i$\hspace{-2mm} & \hspace{-2mm}IIX IXI YII\hspace{-2mm} & \hspace{-2mm}$1 1 1 1 \bar{i}  \bar{i}  \bar{i}  \bar{i} $\hspace{-2mm} & \hspace{-2mm}IIX IXI -YII\hspace{-2mm} & \hspace{-2mm}$1 0 0 i \bar{i}  0 0 1$\hspace{-2mm} & \hspace{-2mm}IXY IYX -YII\hspace{-2mm} & \hspace{-2mm}$1 i 0 0 $-$ i 0 0$\hspace{-2mm} & \hspace{-2mm}IZI YIX -XIZ\hspace{-2mm} & \hspace{-2mm}$1 1 $-$ $-$ \bar{i}  \bar{i}  i i$\hspace{-2mm} & \hspace{-2mm}IIX -IXI -YII\hspace{-2mm} & \hspace{-2mm}$1 \bar{i}  i $-$ \bar{i}  $-$ 1 i$\hspace{-2mm} & \hspace{-2mm}IYZ -IXX -YII\hspace{-2mm} & \hspace{-2mm}$1 i \bar{i}  $-$ 1 \bar{i}  i $-$$\hspace{-2mm} & \hspace{-2mm}XYY -IXX -YIX\hspace{-2mm} & \hspace{-2mm}$1 $-$ $-$ 1 \bar{i}  i i \bar{i} $\hspace{-2mm} & \hspace{-2mm}-IIX -IXI -YII\hspace{-2mm}\hspace{-2mm}\\
\hline\end{tabular}
}
\end{center}

\begin{center}
	\scriptsize
	\scalebox{.6}[.5]{\begin{tabular}{|c|c||c|c||c|c||c|c||c|c||c|c||c|c||c|c|}\hline
\sc State & $\angle$ & \sc State & $\angle$ & \sc State & $\angle$ & \sc State & $\angle$ & \sc State & $\angle$ & \sc State & $\angle$ & \sc State & $\angle$ & \sc State & $\angle$\\ \hline\hline
\hspace{-1mm}$1 1 1 1 1 1 1 1$\hspace{-1mm} & \hspace{-1mm}$\pi/2.59$\hspace{-1mm} & \hspace{-1mm}$1 1 1 1 $-$ $-$ $-$ $-$$\hspace{-1mm} & \hspace{-1mm}$\pi/2.59$\hspace{-1mm} & \hspace{-1mm}$1 0 0 i $-$ 0 0 \bar{i} $\hspace{-1mm} & \hspace{-1mm}$\pi/3$\hspace{-1mm} & \hspace{-1mm}$0 \bar{i}  0 0 1 0 0 0$\hspace{-1mm} & \hspace{-1mm} $\perp$ \hspace{-1mm} & \hspace{-1mm}$1 1 $-$ $-$ $-$ $-$ 1 1$\hspace{-1mm} & \hspace{-1mm}$\pi/2.59$\hspace{-1mm} & \hspace{-1mm}$1 \bar{i}  i $-$ $-$ i \bar{i}  1$\hspace{-1mm} & \hspace{-1mm}$\pi/2.59$\hspace{-1mm} & \hspace{-1mm}$1 0 $-$ 0 0 i 0 i$\hspace{-1mm} & \hspace{-1mm}$\pi/3$\hspace{-1mm} & \hspace{-1mm}$1 $-$ $-$ 1 $-$ 1 1 $-$$\hspace{-1mm} & \hspace{-1mm}$\pi/2.59$\hspace{-1mm}\hspace{-1mm}\\ 
\hspace{-1mm}$1 1 1 1 i i i i$\hspace{-1mm} & \hspace{-1mm}$\pi/2.59$\hspace{-1mm} & \hspace{-1mm}$1 1 1 1 \bar{i}  \bar{i}  \bar{i}  \bar{i} $\hspace{-1mm} & \hspace{-1mm}$\pi/2.59$\hspace{-1mm} & \hspace{-1mm}$1 0 0 i \bar{i}  0 0 1$\hspace{-1mm} & \hspace{-1mm}$\pi/3$\hspace{-1mm} & \hspace{-1mm}$1 i 0 0 $-$ i 0 0$\hspace{-1mm} & \hspace{-1mm}$\pi/3$\hspace{-1mm} & \hspace{-1mm}$1 1 $-$ $-$ \bar{i}  \bar{i}  i i$\hspace{-1mm} & \hspace{-1mm}$\pi/2.59$\hspace{-1mm} & \hspace{-1mm}$1 \bar{i}  i $-$ \bar{i}  $-$ 1 i$\hspace{-1mm} & \hspace{-1mm}$\pi/2.59$\hspace{-1mm} & \hspace{-1mm}$1 i \bar{i}  $-$ 1 \bar{i}  i $-$$\hspace{-1mm} & \hspace{-1mm}$\pi/2.59$\hspace{-1mm} & \hspace{-1mm}$1 $-$ $-$ 1 \bar{i}  i i \bar{i} $\hspace{-1mm} & \hspace{-1mm}$\pi/2.59$\hspace{-1mm}\hspace{-1mm}\\
\hline\end{tabular}
}
\end{center}

\end{document}